\documentclass[noinfoline]{imsart} % Add draft to take out pictures
\usepackage{times}
\usepackage{cite}

\usepackage{amsmath,amssymb,graphicx,bbold,amsthm,float}
\usepackage{subcaption}
\usepackage[linesnumbered,ruled]{algorithm2e}
\usepackage{enumerate}
\usepackage[round]{natbib}
\usepackage[bottom=20mm,left=40mm,right=40mm]{geometry}
\usepackage{xcolor}

\newtheorem{lemma}{Lemma}[section]
\newtheorem{proposition}{Proposition}[section]

\DeclareMathOperator{\argmax}{argmax}
\DeclareMathOperator{\Poisson}{Poisson}
\DeclareMathOperator{\tPoisson}{tPoisson}
\DeclareMathOperator{\Dirichlet}{Dirichlet}
\DeclareMathOperator{\Unif}{Unif}
\DeclareMathOperator{\Ber}{Ber}

\DeclareMathOperator{\tBin}{tBin}
\DeclareMathOperator{\Gam}{Gamma}
\DeclareMathOperator{\Mult}{Mult}

\newcommand{\1}[1]{\mathbb{1}_{#1}}
\newcommand{\sumv}{\sum_{i=1}^n v_i}

\graphicspath{{figures/}{./}}

\begin{document}

\begin{frontmatter}

\title{Nonnegative Bayesian nonparametric factor models with completely random measures for community detection}
\runtitle{Bayesian Nonparametric Poisson factorization}

\begin{aug}
  \author{Fadhel Ayed\ead[label=e1]{fadhel.ayed@stats.ox.ac.uk}}
  \and
  \author{Fran\c cois Caron\ead[label=e2]{caron@stats.ox.ac.uk}}

  \address{Department of Statistics, University of Oxford\\
           \printead{e1,e2}}

  \affiliation{University of Oxford}

  \runauthor{F. Ayed and F. Caron}

\end{aug}

\begin{abstract}
We present a Bayesian nonparametric Poisson factorization model for modeling network data with an unknown and potentially growing number of overlapping communities. The construction is based on completely random measures and allows the number of communities to either increase with the number of nodes at a specified logarithmic or polynomial rate, or be bounded. We develop asymptotics for the number of nodes and the degree distribution of the network and derive a Markov chain Monte Carlo algorithm for targeting the exact posterior distribution for this model. The usefulness of the approach is illustrated on various real networks.
\end{abstract}


\end{frontmatter}

% Introduction

\section{Introduction}
\label{sec:introduction}

Non-negative matrix factorization (NMF) methods~\citep{Paatero1994,Lee2001} aim to find a latent representation of a positive $n\times m$ matrix $A$ as a sum of $K$ non-negative factors. For integer-valued data, Poisson factorization models~\citep{Dunson2005} offer a flexible probabilistic framework for non-negative matrix factorization, and have found wide applicability in signal processing~\citep{Virtanen2008,Cemgil2009} or recommender systems~\citep{Ma2011,Gopalan2015}. In this paper, we focus on the application to network analysis, where $m=n$ and the $n\times n $ count matrix $A$ represents the number of directed or undirected interactions between $n$ individuals; the latent factors may be interpreted as latent and potentially overlapping communities~\citep{Ball2011}, such as sport team members or other social activities circles. We also consider binary data where the matrix represents the existence or absence of a directed or undirected link between individuals. The estimated latent factors can be used for the prediction of missing links/interactions, or for interpretation of the uncovered latent community structure.\smallskip

Poisson factorization approaches require the user to set the number $K$ of latent factors, which is typically assumed to be independent of the sample size $n$. To address this problem, \cite{Zhou2012}, \cite{Gopalan2014} and \cite{Zhou2015} proposed Bayesian nonparametric approaches that allow the number of latent factors to be estimated from the data, and to grow unboundedly with the size $n$ of the matrix. In particular, \cite{Gopalan2014} and \cite{Zhou2015}, considered a Poisson factorization model
\begin{equation}
A_{ij}\sim \Poisson \left  (\sum_{k=1}^\infty r_k v_{ik}v_{jk}\right ),~1\leq i,j\leq n\label{eq:Aij_int}
\end{equation}
where the positive weights $(r_k)_{k\geq 1}$ represent the importance of community $k$, and $v_{ik}>0$ represents the level of affiliation of individual $i$ to community $k$. \cite{Gopalan2014} and \cite{Zhou2015}, extending work from \cite{Titsias2008}, assume that the weights $(r_k)$ are the jumps of a gamma process, ensuring the sum in equation~\eqref{eq:Aij_int} is almost surely finite. Using properties of Poisson random variables, the model~\eqref{eq:Aij_int} can be equivalently represented as
\begin{align}
A_{ij}&=\sum_{k=1}^{\infty} Z_{ijk}\label{eq:Aijsum}\\
Z_{ijk}&\sim \Poisson(r_k v_{ik}v_{jk}),~k=1,2,\ldots\label{eq:Zijk}
\end{align}
for $1\leq i,j\leq n$. The latent count variables $Z_{ijk}$ may be interpreted as the number of latent interactions between two individuals $i$ and $j$ via community $k$, the overall number $A_{ij}$ of interactions being the sum of those community interactions. For example, two members of the same company who also play sport together may meet five times at the company, and twice at the sport center, resulting in seven interactions overall. %\cite{Zhou2015} refers to the partition of $A_{ij}$ in terms of $Z_{ijk}$ as edge-partition model.
The overall number
\begin{equation}
K_n=\sum_{k=1}^\infty \1{\sum_{1\leq i, j\leq n} Z_{ijk}>0}
\end{equation}
of communities $k$ that generated at least one interaction between the $n$ individuals is termed the number of \textit{active communities}. For the gamma process Poisson factor model~\citep{Zhou2015}, the number of active communities $K_n$ grows logarithmically with the number $n$ of individuals. The logarithmic growth assumption may be too restrictive. For example, the number of active communities may actually be unknown but bounded above; alternatively, it may increase at a rate faster or slower than logarithmic. \smallskip

In this paper, we consider generalizations of the gamma process Poisson factorization model, using completely random measures (CRM) ~\citep{Kingman1967}. CRMs offer a flexible and tractable modeling framework~\citep{Lijoi2010}. The proposed models fit in the class of multivariate generalized Indian Buffet process priors recently developed by \cite{James2017} and are also related to compound completely random measures~\citep{Griffin2017}. We consider that $(r_k)$ are the points of Poisson point process with mean measure $\rho$. Depending on the properties of this measure, the number of active communities $K_n$ is either (i) bounded, with a random upper bound, (ii) unbounded and grows sub-polynomially (e.g. $\log n$ or $\log\log n$) or (iii) unbounded and grows as $n^{2 \sigma}$, for some $\sigma \in(0,1)$. For the implementation, we focus in particular on the generalized gamma process~\citep{Brix1999} where a single parameter flexibly controls all three behaviors.\smallskip

The article is organized as follows. In Section \ref{sec:model}, we describe the statistical model for count and binary matrices. The asymptotic properties of the model are derived in Section \ref{sec:properties}. In particular, we relate the asymptotic growth of the number of active features to the regular variation properties of the measure $\rho$. In Section~\ref{sec:inference} we derive a Markov chain Monte Carlo algorithm for posterior inference that does not require any approximation to the original model. In Section~\ref{sec:experiments} we consider applications of our approach to overlapping community detection and link detection in networks, considering real network data with up to tens of thousands of nodes.

\section{Statistical model for count and binary data}
\label{sec:model}

\subsection{General construction}

We present here the model for directed count or binary observations, but the model can be straightforwardly adapted to undirected interactions. Let $(r_k)_{k=1,2\ldots,}$ be the points of a Poisson point process with $\sigma$-finite mean measure $\rho$ on $(0,\infty)$, and assume that $v_{ik}$, $i=1,\ldots,n$, $k\geq 1$, are independent and identically distributed from some probability distribution $F$ on $\mathbb R_+=[0,\infty)$. The variable $v_{ik}$ can be interpreted as the level of affiliation of an individual $i$ to community $k$, and $r_k$ to the importance of that community.

 For count data $(A_{ij})$, where $A_{ij}$ denotes the number of directed interactions from node $i$ to node $j$,  we consider the Poisson factorization model
\begin{equation}
A_{ij}\mid (r_k,v_{ik})\sim \Poisson \left  (\sum_{k=1}^\infty r_k v_{ik}v_{jk}\right ),~1\leq i,j\leq n.\label{eq:Aij}
\end{equation}
Denoting $\Lambda_{ij}=\sum_{k=1}^\infty r_k v_{ik}v_{jk}$ the Poisson rate for $A_{ij}$, the $n\times n$ rate matrix $\Lambda^{(n)}=(\Lambda_{ij})_{1\leq i,j\leq n}$ admits the following factorization as an infinite sum of rank-1 matrices
$$
\Lambda^{(n)}=\sum_{k=1}^\infty r_k v_{1:n,k}v^\intercal_{1:n,k}
$$
where $v_{1:n,k}=(v_{1k},\ldots,v_{nk})^\intercal$. For the model to be well specified, the sum in the right-handside of Equation~\eqref{eq:Aij} needs to be almost surely finite. A necessary and sufficient condition is
\begin{equation}
\iint(1-e^{-rv^2})\rho(dr)F(dv)<\infty~~\text{ and }~~\iint(1-e^{-rv_1 v_2})\rho(dr)F(dv_1)F(dv_2)<\infty.\label{eq:necessary}
\end{equation}
A sufficient set of conditions\footnote{The sufficientness follows from the bound \eqref{eq:bound} given in Appendix.}, which we will assume to hold in the rest of this article, is that $\rho$ is a L\'evy measure and $F$ has finite second moment, that is
\begin{align}
\int_0^\infty(1-e^{-r})\rho(dr)&<\infty~~\text{ and }\tag{A1}\label{eq:assumpt1}\\
\int_0^\infty v^2F(dv)&<\infty.\tag{A2}\label{eq:assumpt2}%\label{eq:sufficient}
\end{align}
In this case, the community affiliations and weights for $n$ nodes can be conveniently represented by a completely random measure
\begin{equation}
G=\sum_{k\geq 1} r_k\delta_{v_{1:n,k}}
\label{eq:CRM}
\end{equation}
on $\mathbb R_+^n$ with mean measure $\rho(dr)F^{\bigotimes^n}(dv_1,\ldots,dv_n)$ where $F^{\bigotimes^n}$ denotes the $n$th product measure of $F$; see \cite{Kingman1967} and \cite{Lijoi2010} for background on CRMs and their applications. If the L\'evy measure is finite, that is, if
$$
\int_0^\infty \rho(dr)<\infty
$$
then the number of points $(r_k)$, and therefore the number of communities, is almost surely finite. Otherwise, when $\int \rho(dr)=\infty$, the number of communities is infinite. \smallskip

When we have binary observation $(Y_{ij})$, we treat the count matrix $(A_{ij})$ as a latent variable, and consider that $Y_{ij}=\mathbb{1}_{A_{ij}>0}$ as in \citep{Caron2017,Zhou2015}. Integrating out $(A_{ij})$, this leads to the following model for binary observations
\begin{equation}
Y_{ij}\mid (r_k,v_{ik})\sim \Ber \left (1-\exp \left  [\sum_{k=1}^\infty r_k v_{ik}v_{jk}\right ]\right ),~1\leq i,j\leq n.
\label{eq:Yij}
\end{equation}

\subsection{Specific model}

In the inference and experimental part, we use the following choice for the $\rho$ and $F$. The L\'evy measure $\rho$ is taken to be that of a generalized gamma process (GGP, see \cite{Hougaard1986}, \cite{Brix1999})
\begin{equation}
\rho(dr) = \frac{\kappa}{\Gamma (1-\sigma_0)} r^{-1-\sigma} e^{-\tau r} dr
\end{equation}
where $\sigma_0\in (-\infty, 1)$, $\kappa>0$ and $\tau>0$. When $\sigma_0=0$, we obtain a gamma process, and the model corresponds to that of \cite{Zhou2015}. When $\sigma_0 < 0$, the L\'evy measure is finite, while when $\sigma_0 \geq 0$, the L\'evy measure is infinite.

Concerning the affiliations, we will assume that $F$ is a gamma distribution with parameters $\alpha>0$ and $\beta>0$. That is, the probability density function (pdf) $f$ is given by

$$ f(v) = \frac{\beta^\alpha}{\Gamma(\alpha)} v^{\alpha-1} e^{-\beta v}$$
where $\Gamma$ denotes the usual gamma function. The hyperparameters $(\kappa,\sigma_0,\tau,\alpha,\beta)$ and $(\kappa^\prime=\kappa/ \beta^{2\sigma_0},$ $\sigma_0, \tau^\prime=\tau\beta^2,\alpha,1)$ induce the same distribution for the latent factors $(\Lambda_{ij})$. In order to guarantee the identifiability of the hyperparameters, we therefore set $\beta=1$.

\subsection{Related work}
\label{sec:discussion}
The model introduced in this section can be seen from different perspectives that nicely connect it to the existing literature. First, the model can be seen as obtained from a functional of a CRM. Recall the definition of the CRM $G$ in Eq. \eqref{eq:CRM}. Define the $n\times n$ matrix $\Lambda^{(n)}$ as the following functional of $G$
$$
\Lambda^{(n)}=\int_{(0,\infty)^n} h(u)G(du)=\sum_{k\geq 1} r_k v_{1:n,k} v^\intercal_{1:n,k}
$$
where $h(u)=u u^\intercal$.% Then $A_{ij}|G\sim \Poisson (\Lambda_{ij})$. 
Alternatively, this can be interpreted in the framework of compound completely random measures~\citep{Griffin2017}. For each $1\leq i,j\leq n$, denote $G_{ij}=\sum_{k\geq 1} r_k v_{ik}v_{jk}\delta_{\zeta_{k}}$ where $\zeta_k$ are some community locations in some domain $\Theta$, iid from some distribution $H$, irrelevant here. Then $(G_{ij})_{1\leq i,j\leq n}$ are compound CRMs on $\Theta$ and $\Lambda_{ij}=G_{ij}(\Theta)$. In the same vein, the model can also be interpreted as an instance of the class introduced by~\citep[Section 5]{james2014poisson}. Denote $Z^{(n)}_k$  the $n\times n$ matrix with entries $Z_{ijk}$. Then the matrix-valued process $\sum_{k\geq 1} Z_k^{(n)}\delta_{\zeta_k}$ is a draw from a multivariate Indian buffet process.

Finally, as mentioned in the introduction, the model admits as a special case the Poisson factorization based on the gamma process of \cite{Zhou2015}.

\section{Asymptotic Properties}
\label{sec:properties}

In this section we study the asymptotic properties of the proposed class of models, and in particular the growth rate of the number of active communities as the sample size $n$ grows, and the asymptotic proportion of communities of a given size. For a given sequence $(r_k)_{k\geq 1}$ and $(v_{ik})_{i\geq 1,k\geq 1}$, denote $A^{(n)}_{ij}$ and $Z^{(n)}_{ijk}$ where $n\geq 1$, $1\leq i,j\leq n$, $k\geq 1$ respectively the number of directed interactions and the number of community directed interactions distributed from Equations \eqref{eq:Aijsum} and \eqref{eq:Zijk}. We consider two different asymptotic settings
\begin{itemize}
\item Constrained setting. For any $1\leq m \leq n$, and $1\leq i,j \leq m$, $A_{ij}^{(n)}=A_{ij}^{(m)}$. In this setting, we suppose that the connections between the already observed nodes remain unchanged. It is equivalent to assuming that there is an infinite but  fixed graph and $A^{(n)}$ represents the connections between the $n$ first nodes of that graph.
\item Unconstrained setting. This setting is more general, and we only assume that $A^{(n)}_{ij}$ and $Z^{(n)}_{ijk}$ are marginally sampled from Equations \eqref{eq:Aijsum} and \eqref{eq:Zijk}.
\end{itemize}

All the results of this section, otherwise stated, hold for the unconstrained setting. We indicate when a stronger result holds in the constrained setting. All proofs are given in Appendix~\ref{sec:proofs}.

\subsection{General model}

Let $d^{(n)}_k$ be the degree of the community/feature $k$, corresponding to the number of interactions amongst $n$ individuals due to community $k$, and defined as
\begin{equation} d^{(n)}_k = \sum\limits_{1\leq i,j\leq n} Z^{(n)}_{ijk}.
\end{equation}

A community is active if $d^{(n)}_k\geq 1$. The number of active communities is therefore defined as
\begin{equation}
K_n = \sum\limits_{k=1}^\infty \1{d^{(n)}_k \geq 1}
\end{equation}
Denote $K_{n,j}$ the number of communities with degree $j\geq 1$
$$ K_{n,j} = \sum\limits_{k=1}^\infty \1{d^{(n)}_k = j}$$

Note that under the constrained setting, $d^{(n)}_{k}$, $K_n$ and $\sum_{\ell \geq j}K_{n,\ell}$ are all almost surely increasing with the sample size $n$, whereas this is not necessarily the case for the unconstrained setting.

\begin{proposition} \label{lemma:Kpoisson}
Under Assumptions \eqref{eq:assumpt1} and \eqref{eq:assumpt2}, the number of active communities $K_n$ is a Poisson random variable with mean
\begin{equation}
\Psi(n) = \int \left (1-e^{-r (\sum_{i=1}^n v_i)^2 } \right ) \left [\prod\limits_{i=1}^n F(dv_i) \right ]\ \rho(dr) < \infty.
\label{eq:Psi}
\end{equation}
  The number $K_{n,j}$ of communities with degree $j$ is also Poisson distributed, with mean
  \begin{equation}
  \Psi_j(n)=\frac{1}{j!} \int r^j \left (\sumv\right )^{2j} e^{-r \left (\sumv\right )^2 } \left [\prod\limits_{i=1}^n F(dv_i) \right ]\ \rho(dr).
\end{equation}

  Finally, for $j\geq 1$, $\sum\limits_{k \geq j} K_{n,k}$, the number of communities with degree at least $\ell$, is also Poisson distributed with mean $\sum\limits_{k \geq j} \Psi_k(n)$.
\end{proposition}

% Proof at end

In the rest of the section we relate the asymptotic behavior of quantities of interest to the properties of the mean measure $\rho$. Let consider the tail L\'evy intensity defined as
$$ \forall x > 0,\ \overline{\rho}(x) = \int_{x}^{\infty} \rho(dr).$$
We assume that $\overline{\rho}$ is a regularly varying function at 0, that is
\begin{equation}\ \overline{\rho}(x) \sim x^{-\sigma} \ell(1/x)\text{ as }x\rightarrow 0 \tag{A4}\label{assumption:reg_var}
\end{equation}
where $\sigma\in[0,1)$ and $\ell$ is a slowly varying function verifying $\lim_{t \rightarrow +\infty} \ell(at)/\ell(t)  = 1$ for all $a>0$. Examples of slowly varying functions include functions converging to a constant, $\log^a t$ for any $t$, $\log\log t$, etc. Note that the CRM is finite activity if and only if $\sigma=0$ and $\ell(t)\rightarrow C<\infty$.

Now, let us consider the asymptotic behavior of the number of active communities $K_n$.

\begin{proposition}\label{lemma:asymptotic_K}
Let $K_n$ be the number of active communities.
Then for $0\leq \sigma < 1$,
\begin{equation}\label{expectation_K_asympt}
\mathbb{E}[K_n] \sim \Gamma(1-\sigma) m_f^{2\sigma} n^{2\sigma} \ell(n^2)
\end{equation}
as $n$ tends to infinity, where $m_f = \int v F(dv)$. Additionally, for $0 < \sigma < 1$,
\begin{equation}\label{K_asympt}
K_n \sim \mathbb{E}[K_{n}] \ \ \text{a.s.}
\end{equation}
If we further assume that the sequence $(K_n)_{n\geq 1}$ is almost surely non-decreasing (as in the constrained setting), then $(\ref{K_asympt})$ holds for $\sigma = 0$ and $\ell(t)\rightarrow\infty$ as well. In the finite activity case, that is $\sigma=0$ and $\ell(t)\rightarrow \overline \rho(0)=\int_0^\infty \rho(dr)<\infty$, we have
\begin{align*}
K_n\rightarrow K_\infty
\end{align*}
as $n$ tends to infinity, where $K_\infty$ is a Poisson random variable with mean $\overline \rho(0)$. The above convergence holds in distribution for the unconstrained setting and almost surely for the constrained setting.
\end{proposition}

\begin{proposition}\label{degree:as}

Let $K_{n,j}$ be the number of communities of degree $j$. Then for $ 0 < \sigma < 1$ and any $j\geq 1$,

\begin{equation}
K_{n,j} \sim \frac{\sigma \Gamma(j-\sigma) }{j!} m_f^{2\sigma} n^{2 \sigma} \ell(n^2)\ \ \text{a.s.}
\end{equation}
as $n$ tends to infinity. Therefore,
\begin{equation}\label{prop_as}
\frac{K_{n,j}}{K_n} \rightarrow \frac{\sigma \Gamma(j-\sigma)}{ \Gamma(1-\sigma) j!}\ \ \text{a.s.}
\end{equation}
as $n$ tends to infinity. This corresponds to a power-law behavior as
$$
\frac{\sigma \Gamma(j-\sigma)}{ \Gamma(1-\sigma) j!} \sim \frac{\sigma}{j^{\sigma+1}}
$$
for large $j$. If we further assume that for all $k \geq 1$, $\left (\sum\limits_{j \geq k} K_{n,j}\right )_{n\geq 1}$ is non-decreasing (constrained setting), then (\ref{prop_as}) holds also for $\sigma = 0$ and $\ell(t)\rightarrow\infty$.
\end{proposition}

Finally, let $c^{(n)}(k,k')$ denote the cosine between the corresponding affiliation vectors
$$ c^{(n)}(k,k') = \frac{\sum_{i=1}^n v_{ik}v_{ik'}}{\sqrt{\sum_i v_{ik}^2}\sqrt{\sum_i v_{ik'}^2}} .$$
This coefficient gives a measure of the overlap between two communities $k$ and $k'$. By the law of large numbers, for any $k\neq k'$,
$$ c^{(n)}(k,k'){\sim} \frac{(\int v F(dv))^2}{\int v^2 F(dv)} \ \ \textit{ a.s.  as }n\rightarrow \infty.$$

\subsection{Specific case of the GGP}

In the case of the GGP, we have
\begin{align*}
\overline\rho (x)=\frac{\kappa\tau^{\sigma_0}\Gamma(-\sigma_0,\tau x)}{\Gamma(1-\sigma_0)}\sim \left \{ \begin{array}{ll}
                        -\frac{\kappa\tau^{\sigma_0}}{\sigma_0} & \text{if }\sigma_0<0 \\
                        \kappa\log(1/x) & \text{if }\sigma_0=0 \\
                        \frac{\kappa x^{-\sigma_0}}{\sigma_0\Gamma(1-\sigma_0)} & \text{if }\sigma_0>0
                      \end{array}\right .
\end{align*}
as $x$ tends to $0$, where $\Gamma(a,x)$ is the incomplete gamma function. Note that $\overline\rho (x)$ is of the form $x^{-\sigma}\ell(1/x)$ where $\sigma=\max(0,\sigma_0)$ and
\begin{align*}
\ell(t)=\left \{ \begin{array}{ll}
                        -\frac{\kappa\tau^{\sigma_0}}{\sigma_0} & \text{if }\sigma_0<0 \\
                        \kappa\log(t) & \text{if }\sigma_0=0 \\
                        \frac{\kappa }{\sigma_0\Gamma(1-\sigma_0)} & \text{if }\sigma_0>0
                      \end{array}\right .
\end{align*}
is a slowly varying function at infinity. The results of the previous subsection therefore apply. For simplicity, we state the results for the constrained setting. We have, almost surely as $n\rightarrow\infty$
\begin{align*}
K_n\sim \left \{ \begin{array}{ll}
                        K_\infty & \text{if }\sigma_0<0 \\
                        2\kappa\log(n) & \text{if }\sigma_0=0 \\
                        \kappa\alpha^{2\sigma_0}n^{2\sigma_0} /\sigma_0 & \text{if }\sigma_0>0
                      \end{array}\right .
\end{align*}
where $K_\infty\sim \Poisson(-\kappa\tau^{\sigma_0}/\sigma_0)$. Additionally, for $\sigma\geq 0$,
\begin{align*}
\frac{K_{n,j}}{K_n}\rightarrow \frac{\sigma_0 \Gamma(j-\sigma_0)}{ \Gamma(1-\sigma_0) j!}
\end{align*}
almost surely as $n\rightarrow\infty$. Finally,
$$ c^{(n)}(k,k')\rightarrow\frac{\alpha}{\alpha +1}.$$

Therefore, $\sigma_0$ governs the asymptotic behavior of the number of active communities. $K_n$ is bounded with a random upper bound ($\sigma_0<0$), increases logarithmically ($\sigma_0=0$) or polynomially ($\sigma_0>0$). In the polynomial case, $\sigma_0$ also controls the power-law exponent of the proportion of communities of a given size. The parameter $\kappa$ is an overall linear scaling parameter. Finally, the parameter $\alpha$ governs the amount of  overlapping between two communities.

\section{Simulation, posterior characterization and inference}
\label{sec:inference}

In this section we describe the marginal distribution and conditional characterization of the model. Building on these, we derive an exact sampler for simulating from the model, and a Markov chain Monte Carlo algorithm to approximate the posterior distribution. Importantly, the sampler targets the distribution of interest and does not require any truncation or approximation. For simplicity of exposition, we assume that $\rho$ and $F$ are absolutely continuous with respect to the Lebesgue measure, with $\rho(dr)=\rho(r)dr$ and $F(dx)=f(x)dx$.
\subsection{Marginal distribution and simulation}

For a fixed $n$, recall that $K_n$ denotes the number of active communities. Let $((\widetilde r_1,\widetilde v_{1:n,1}),\ldots,$ $(\widetilde r_{K_n},\widetilde v_{1:n,K_n}))$ be the subsequence of $(r_k,v_{1:n,k})$ such that community $k$ is active, meaning that $\sum_{1\leq i,j\leq n} Z_{ijk}\geq 1 $, arranged in random order. Let $\widetilde Z_{ijk}$ be the number of community interactions corresponding to the active community $(\widetilde r_k, \widetilde v_{1:n,k})$. Note that
\begin{equation}
A_{ij}=\sum_{k=1}^{K_n}\widetilde Z_{ijk}.\label{eq:AijZtilde}
\end{equation}
Let $\widetilde Z_k=(\widetilde Z_{ijk})_{1\leq i,j\leq n}$. Using Proposition 5.2 of \cite{James2017}, we obtain the following lemma.

\begin{lemma}[Marginal distribution]
\label{th:marginal}
The joint distribution of $(K_n, (\widetilde r_{1:K_n}, \widetilde v_{1:n,1:K_n}), (\widetilde Z_k)_{k=1,\ldots,K_n})$ is given by
\begin{align}
K_n &\sim \Poisson(\Psi(n))\label{eq:KnPoisson}
\end{align}
where $\Psi(n)$ is defined in Eq.\eqref{eq:Psi}, and
\begin{align*}
p((\widetilde r_{1:K_n}, \widetilde v_{1:n,1:K_n})|K_n)&=\prod_{k=1}^{K_n} p(\widetilde r_k,\widetilde v_{1:n,k}|K_n)
\end{align*}
where
\begin{align}
p(\widetilde r_k,\widetilde v_{1:n,k}|K_n)\propto (1-e^{-\widetilde r_k (\sum_{i=1}^n \widetilde v_{ik})^2})\rho(\widetilde r_k)\prod_{i=1}^n f(\widetilde v_{ik}).
\label{eq:conditionalrv}
\end{align}
Finally, for each $k=1,\ldots,K_n$,
\begin{align}
\widetilde Z_k|(\widetilde r_{1:K_n}, \widetilde v_{1:n,1:K_n})&\sim \tPoisson(\widetilde r_k \widetilde v_{1:n,k}\widetilde v^\intercal_{1:n,k} )\label{eq:ZktruncPoisson}
\end{align}
where $\tPoisson(\Lambda)$ denotes the distribution of a integer-valued matrix with Poisson entries with mean values $\Lambda_{ij}$, conditionally on the sum of the entries being strictly positive. This has probability mass function
$$p(A)=
\left \{
  \begin{array}{ll}
    (1-e^{-\sum_{ij} \Lambda_{ij}})^{-1}\prod_{1\leq i,j\leq n} \frac{\Lambda_{ij}^{A_{ij}}e^{-\Lambda_{ij}}}{A_{ij}!} & \text{if }\sum_{ij} A_{ij}>0 \\
    0 & \text{otherwise}\\
  \end{array}
\right.
$$
\end{lemma}

The model has an infinite number of parameters, but Lemma~\ref{th:marginal} allows us to derive an algorithm to exactly sample from it, by successively simulating $K_n$, $(\widetilde r_{1:K_n}, \widetilde v_{1:n,1:K_n})$, $(\widetilde Z_k)_{k=1,\ldots,K_n}$ and $A$ using Equations \eqref{eq:KnPoisson}, \eqref{eq:conditionalrv}, \eqref{eq:ZktruncPoisson} and \eqref{eq:AijZtilde}.

Sampling from the conditional distribution \eqref{eq:ZktruncPoisson} can be done efficiently by first sampling the number of multiedges $\sum_{i,j} \widetilde Z_{i,j,k}$ from a truncated Poisson with mean $\widetilde r_k (\sum_i \widetilde v_{i,k})^2$, then sampling iid the end nodes of the edges proportionally to the affiliation vector. Simulating from the conditional distribution  \eqref{eq:conditionalrv} can be more challenging since it requires sampling a $n+1$ dimensional vector. However, if we suppose that the affiliations are Gamma distributed,   the problem reduces to sampling $(\widetilde r_k,\sum_i \widetilde v_{i,k})$, which is a two dimensional vector, and independently sample the normalized affiliations from a Dirichlet distribution. Indeed, if the affiliations are Gamma distributed, we consider the following change of variable.
\begin{align}
\widetilde \varsigma_k&=\sum_{i=1}^n \widetilde v_{ik},~~~k=1,\ldots,K_n\\
\widetilde \varphi_{ik}&=\frac{\widetilde v_{ik}}{\widetilde \varsigma_k},~~~~k=1,\ldots,K_n;~i=1,\ldots,n
\end{align}
This gives the following algorithm for exact simulation from the model.
\begin{enumerate}
\item Sample $K_n$ from Eq. \eqref{eq:KnPoisson}
\item For $k=1,\ldots,K_n$
\begin{enumerate}
\item Sample $(\widetilde \varphi_{1k},\ldots,\widetilde \varphi_{nk})\sim \Dirichlet(\alpha,\ldots,\alpha)$
\item Sample $\widetilde \varsigma_{k}$ from
\begin{equation}
p(\widetilde\varsigma)\propto \psi(\widetilde\varsigma^2)\Gam(\widetilde\varsigma;n\alpha,\beta)\label{eq:varsigma}
\end{equation}
\item Sample $\widetilde r_{k}|\widetilde \varsigma_{k}$ from
\begin{equation}
p(\widetilde r\mid \widetilde \varsigma)\propto (1-e^{-\widetilde r\widetilde \varsigma^2}) \rho(\widetilde r)\label{eq:vartilder}
\end{equation}
\item Sample $\widetilde Z_k^{(n)}$ from Eq. \eqref{eq:ZktruncPoisson}
\end{enumerate}
\item For $1\leq i,j\leq n$, set $A_{ij}=\sum_{k=1}^{K_n} \widetilde Z_{ijk}$
\end{enumerate}
where $\psi(t)=\int_0^\infty (1-e^{-wt})\rho(dw)$ is the Laplace exponent, $\Dirichlet(\alpha,\ldots,\alpha)$ denotes the standard Dirichlet distribution and $\Gam(x;a,b)$ denotes the probability density function of a Gamma random variable with parameters $a$ and $b$, evaluated at $x$. In the case of the GGP, the Laplace exponent is
\begin{equation}
\psi(t)=\frac{\kappa}{\sigma}((t+\tau)^\sigma - \tau^\sigma).\label{eq:LaplaceExpoGGP}
\end{equation}
One can sample from Eq. \eqref{eq:varsigma} and \eqref{eq:vartilder} using rejection.

\subsection{Posterior characterization}

Using Proposition 5.1 in \citep{James2017}, one can characterize the conditional distribution of the CRM $G$ given the latent community counts $\widetilde Z_{ijk}$.

\begin{lemma}\label{posterior}
Conditionally on $(\widetilde Z_{k}^{(n)})_{k=1,\ldots,K_n}$, the CRM $G$ has the same distribution as
$$G^\prime + \sum\limits_{k = 1}^{K_n} \tilde{r}_k \delta_{\tilde{v}_{1:n,k}}$$
where $G^\prime$ is an inhomogeneous CRM on $\mathbb R_+^n$ with mean intensity $$e^{-r (\sum\limits_{i=1}^n v_{i})^2} \rho(r)\prod_{i=1}^n f(v_i)$$ and $(\tilde{r}_k,\tilde{v}_{1:n,k})_{k=1,\ldots,K_n}$ are independent of $G^\prime$ and iid with density
\begin{align}
p(\tilde{r}_k,\tilde{v}_{1:n,k}|\widetilde Z_k^{(n)})=e^{-\widetilde r_k (\sum_i \widetilde v_{ik})^2 } \widetilde r_k^{\widetilde d_k}\rho(\widetilde r_k) \prod_{i=1}^n \widetilde v_{ik}^{\widetilde m_{ik}}f(\widetilde v_{ik})
\end{align}
where $\widetilde m_{ik}=\sum_j \widetilde Z_{ijk}+\widetilde Z_{jik}$ and $\widetilde d_k=\sum_{i,j} \widetilde Z_{ijk}$.

\end{lemma}

In the case where $f$ is a gamma pdf, we can use the same reparameterization as in the previous subsection with $(\widetilde \varsigma_k, \widetilde \varphi_{1:n,k})$ in place of $\widetilde v_{1:n,k}$. This leads to the following conditional distributions.
\begin{align*}
\widetilde \phi_{1:n,k}|\widetilde Z_k^{(n)} &\sim\Dirichlet(\alpha+\widetilde m_{1k},\ldots,\alpha+\widetilde m_{nk})\\
p(\widetilde \varsigma_k|\widetilde Z_k^{(n)})&\propto\varkappa(\widetilde d_k,\widetilde \varsigma^2) \Gam(\widetilde \varsigma;n\alpha+2\widetilde d_k,\beta) \\
p(\widetilde r_k|\widetilde \varsigma_k,\widetilde Z_k^{(n)})&\propto e^{-\widetilde r_k\widetilde \varsigma_k^2}\widetilde r_k^{\widetilde d_k}\rho(\widetilde r_k)
\end{align*}
where $\varkappa(m,t)=\int_0^\infty r^m e^{-rt}\rho(r)dr$. In the GGP case, we have
\begin{align*}
\boldsymbol\varkappa(m,t)=\kappa\frac{\Gamma(m-\sigma)}{\Gamma(1-\sigma)}(t+\tau)^{\sigma-m}
\end{align*}
and
$$
\widetilde r_k|\widetilde \varsigma_k,\widetilde Z_k^{(n)}\sim \Gam(\widetilde d_k-\sigma, \widetilde \varsigma_k^2+\tau).
$$

\subsection{Slice sampler for posterior inference}

We recall that $\theta$ denote the set of hyperparameters of the mean measure $\rho$ and pdf $f$. To simplify the presentation, here we suppose that we observe the complete adjacency matrix $A$, which means that we observe a directed and weighted graph with no missing (hidden) edge. The objective is obtain samples distributed from the conditional distribution
$$
p(K_n,(\widetilde r_k,\widetilde v_{1:n,k})_{k=1,\ldots,K_n},\theta \mid A).
$$
In the Appendix, we show how to do inference when we only observe a partial graph (with missing edges to predict) that can be directed or undirected, weighted or binary. In order to leverage the Poisson factorization construction, we augment the model with the latent community counts $\widetilde Z_k$. Additionally, to deal with the unknown number of active communities $K_n$, we use auxiliary slice variables, similarly to other Gibbs sampler for Bayesian nonparametric models~\citep{Walker2007,Kalli2011,Favaro2013}. For each directed pair $(i,j)$ such that $A_{ij}\geq 1$ consider the scalar latent variable
\begin{equation}
s_{ij}|(\widetilde r_k,\widetilde Z_{ijk})_{k=1,\ldots,K_n}\sim \Unif\left (0, \min_{\{k|\widetilde Z_{ijk}\geq 1\}} \widetilde r_k\right )
\end{equation}
and denote $s=\min_{ij} s_{ij}$. Note that by definition, $\widetilde r_k\geq s$ for all $k=1,\ldots,K_n$. Let
\begin{align*}
\overline G &=\sum_{k} r_k\delta_{v_{1:n,k}}\1{r_k\geq s}:=\sum_{k=1}^{\overline K_n} \overline r_k\delta_{\overline v_{1:n,k}}
\end{align*}
 be the CRM corresponding to the set of active or inactive communities with weight $r_k\geq s$, of (almost surely finite) cardinality $\overline K_n\geq K_n$. Denote $\overline Z_{ijk}\geq 0$ the associated community interactions, and $\overline Z_{k}=(\overline Z_{ijk})$. The data augmented slice sampler draws samples asymptotically distributed from
$$
p((\overline Z_k)_{k=1,\ldots,\overline K_n},\overline G,\theta, s \mid A).
$$

The main steps of the algorithm are as follows.
\begin{enumerate}
\item For each directed pair $(i,j)$ such that $A_{ij}\geq 1$, Update $(\overline Z_{ijk})_{k=1,\ldots,\overline K_n}$ given the rest of the variables,
\item Update the hyperparameters $\theta$ given the rest  of the variables,
\item Update $(\overline G, s)$ given the rest of the variables.
\end{enumerate}
The details of each step are given in Appendix~\ref{sec:appendixGibbs}. Each iteration of the Gibbs sampler has a time complexity scaling in $\overline K_n S$ where $S$ is the number of nonzero entries of the matrix. Therefore, the algorithm takes advantage of the sparsity of the networks. Additionally, each entry of the sparse graph can be dealt with independently, making the algorithm straightforwardly parallelizable.

\section{Experiments}
\label{sec:experiments}

We implement the algorithm described in the previous section with the GGP-Gamma scores model. We assign Gamma priors on the hyperparameters $\kappa, \tau, \alpha$ with parameters $(0.1,0.1)$. We fix $\beta = 1$. We allow up to a linear growth of the number of communities, corresponding to $\sigma<  0.5$, for small datasets and use a Gamma prior with parameter $(0.1,0.1)$ on $1-2\sigma$. For larger datasets, we restrict $\sigma < 0.25$, meaning that the number of communities cannot grow at a faster rate than $\sqrt{n}$. This is obtained by using a Gamma prior with parameter $(0.1,0.1)$ on $1-4\sigma$.

\subsection{Synthetic dataset}
We first run the algorithm on a synthetic dataset simulated from our model, to check that the algorithm can recover the true parameters. We sample a directed and unweighted graph from the GGP-gamma model with size $n=800$ and $\sigma = 0.2,\kappa=1,\tau=0.15,\alpha=0.05,\beta=0.2$. The number of edges of the obtained graph is $20198$, and the true number of active communities is $42$.
We run three chains in parallel with $500,000$ iterations, with $250,000$ iterations for burn-in. We show in Figure~\ref{fig:synt} trace plots of the number of active communities $K_n$ and parameter $\sigma$ showing the MCMC algorithm can recover these parameters.

\begin{figure}[H]
\begin{subfigure}{.5\textwidth}
  \centering
  \includegraphics[width=.8\linewidth]{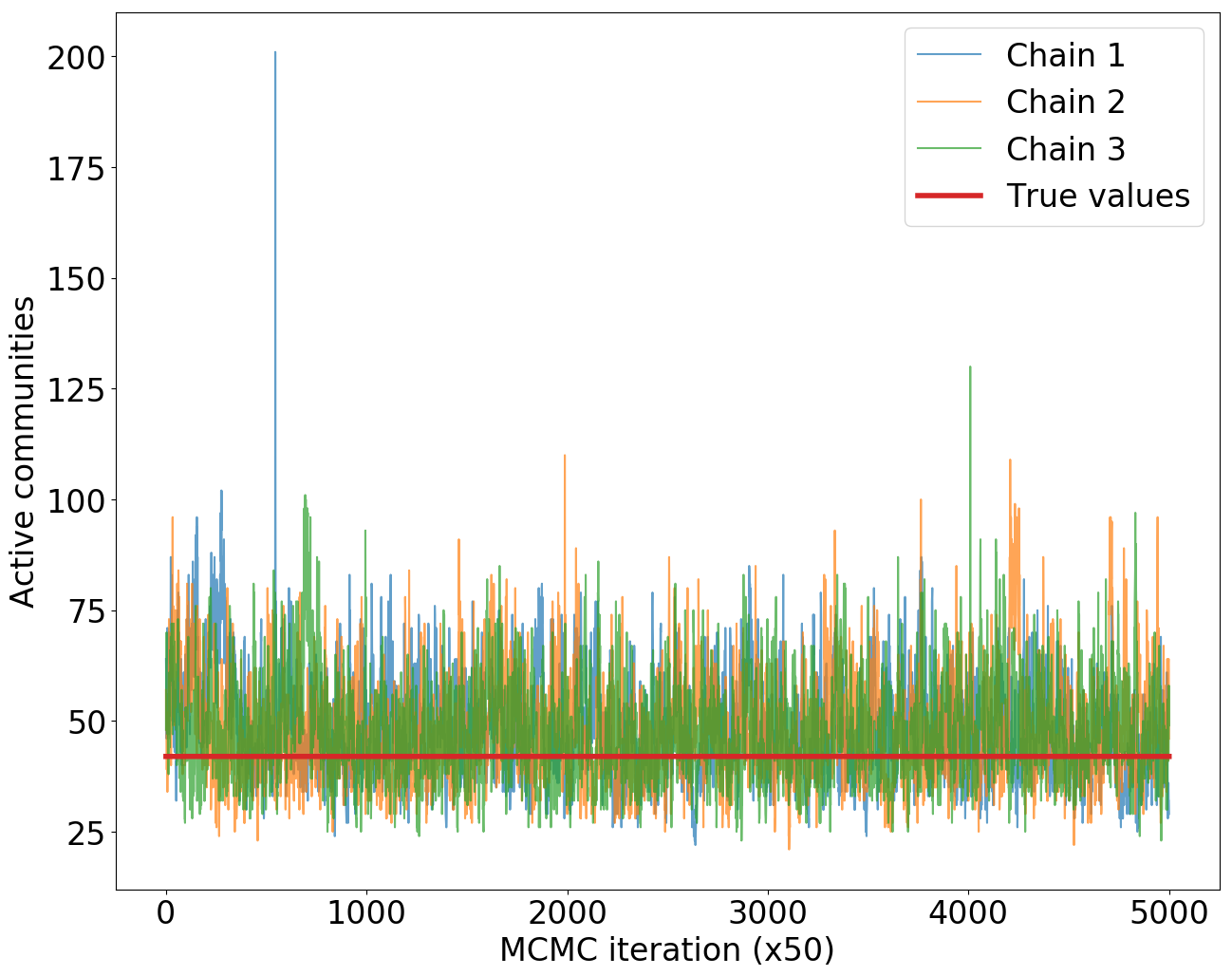}
  \caption{Trace plot of $K_n$}
  \label{fig:synt_active_co}
\end{subfigure}%
\begin{subfigure}{.5\textwidth}
  \centering
  \includegraphics[width=.8\linewidth]{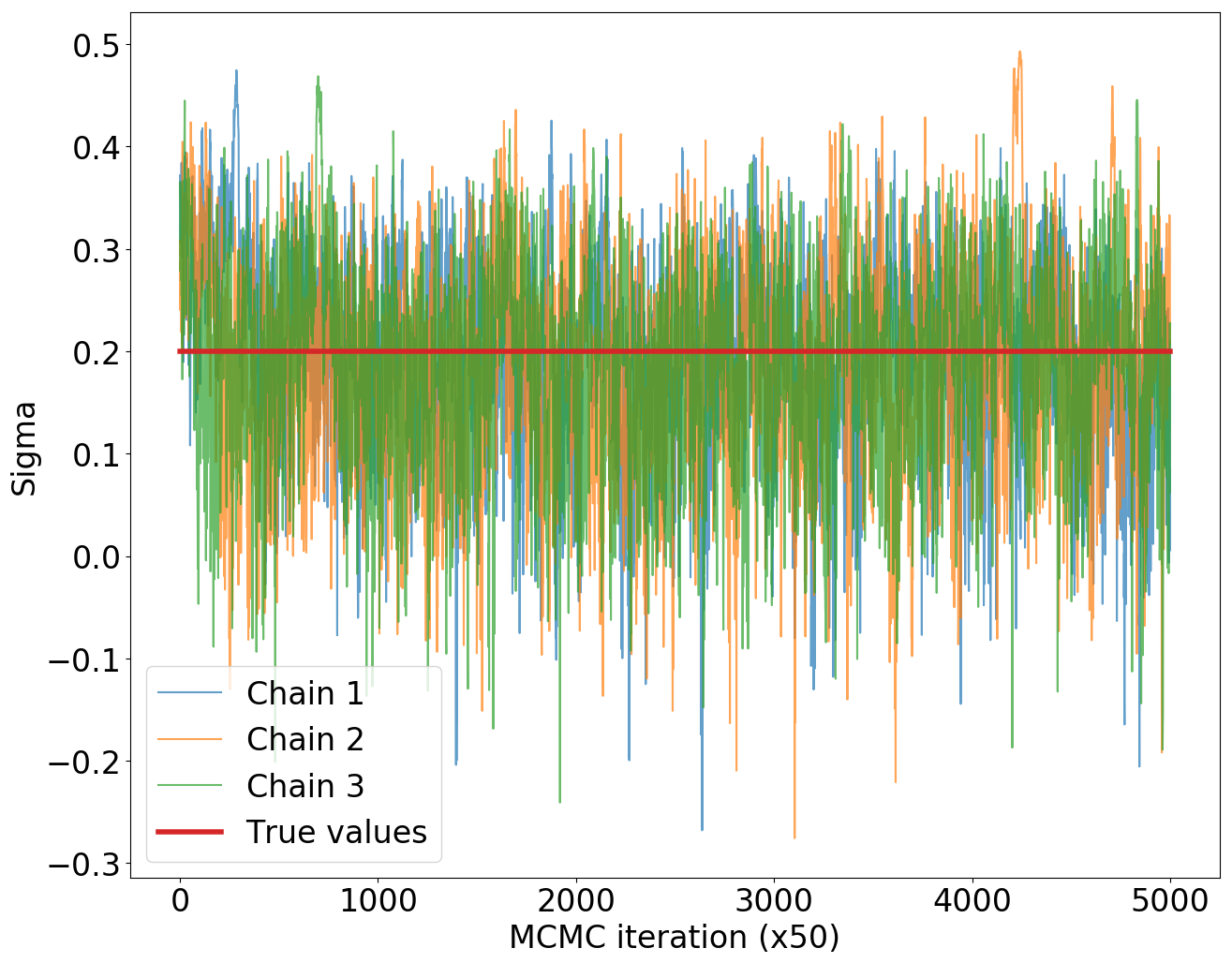}
  \caption{Trace plot of $\sigma$}
  \label{synt:synt_sigma}
\end{subfigure}
\caption{Trace plots of (a) the number of active communities $K_n$ and (b) $\sigma$, on a synthetic example.}
\label{fig:synt}
\end{figure}

\subsection{Political blogs}
The polblogs network~\citep{Adamic2005} is the network of the American political blogosphere in February 2005. It is a directed unweighted graph, where there is an edge $(i,j)$ if blog $i$ cites blog $j$. It is composed of $1490$ nodes and $19025$ edges. For each node, some ground truth information about its political affiliation (republican/democrat) is known. 

We will use this dataset in order to illustrate the role of the parameter $\alpha$ in the model. As indicated in Section~\ref{sec:properties} this parameter tunes the amount of overlapping between the communities. A smaller value enforces less overlap between communities. We run three chains with $500,000$ iterations. The posterior samples of $\sigma$ for three different values of $\alpha$ are in also shown in Figure~\ref{fig:polblogs}. The model allows overlapping communities but, for visualization purposes, it is useful to obtain an associated partition of the nodes. For each iteration, one can cluster the nodes by assigning each node to the community where it is most active. That is, at iteration $t$ of the MCMC algorithm, define for $i=1,\ldots,n$
$$
c_i^{(t)}=\argmax_{k} \{\sqrt{r^{(t)}_k} v^{(t)}_{ik}\}
$$
the cluster membership of node $i$. We then compute an approximate Bayesian point estimate $\widehat c=(\widehat c_1,\ldots,\widehat c_n)$ of the partition of the nodes, using Binder's loss function~\citep{Lau2007}. Nodes are reordered according to their estimated membership $\widehat c$, and Figure~\ref{fig:polblogs} shows the densities of connection between and within clusters for three different values of $\alpha$.  Depending on the amount of overlapping, we obtain two ($\alpha=0.8$), three ($\alpha=0.4$) or four ($\alpha=0.2$) communities. In order to interpret those communities, we calculate in Table~\ref{tab:polblogs} for each community the proportion of interactions between democrat blogs, between a democrat and a republican blog, and between two republican blogs. For $\alpha=0.8$, there are two estimated communities which can clearly be identified as democrat (community \#1) and republican (community 2). For $\alpha=0.4$, we have three communities. One is mostly associated to democrat blogs (\#1) while the other two correspond to a split of the republican blogs into right (\#2) and center-right (\#3) groups. For $\alpha=0.2$, we obtain a further split of the democrat blogs into left (\#1) and center-left (\#2) groups. Increasing the value of $\alpha$ therefore leads to a finer and finer partition of the nodes.

\begin{figure}[h]
\begin{subfigure}{.5\textwidth}
  \centering
  \includegraphics[width=.6\linewidth]{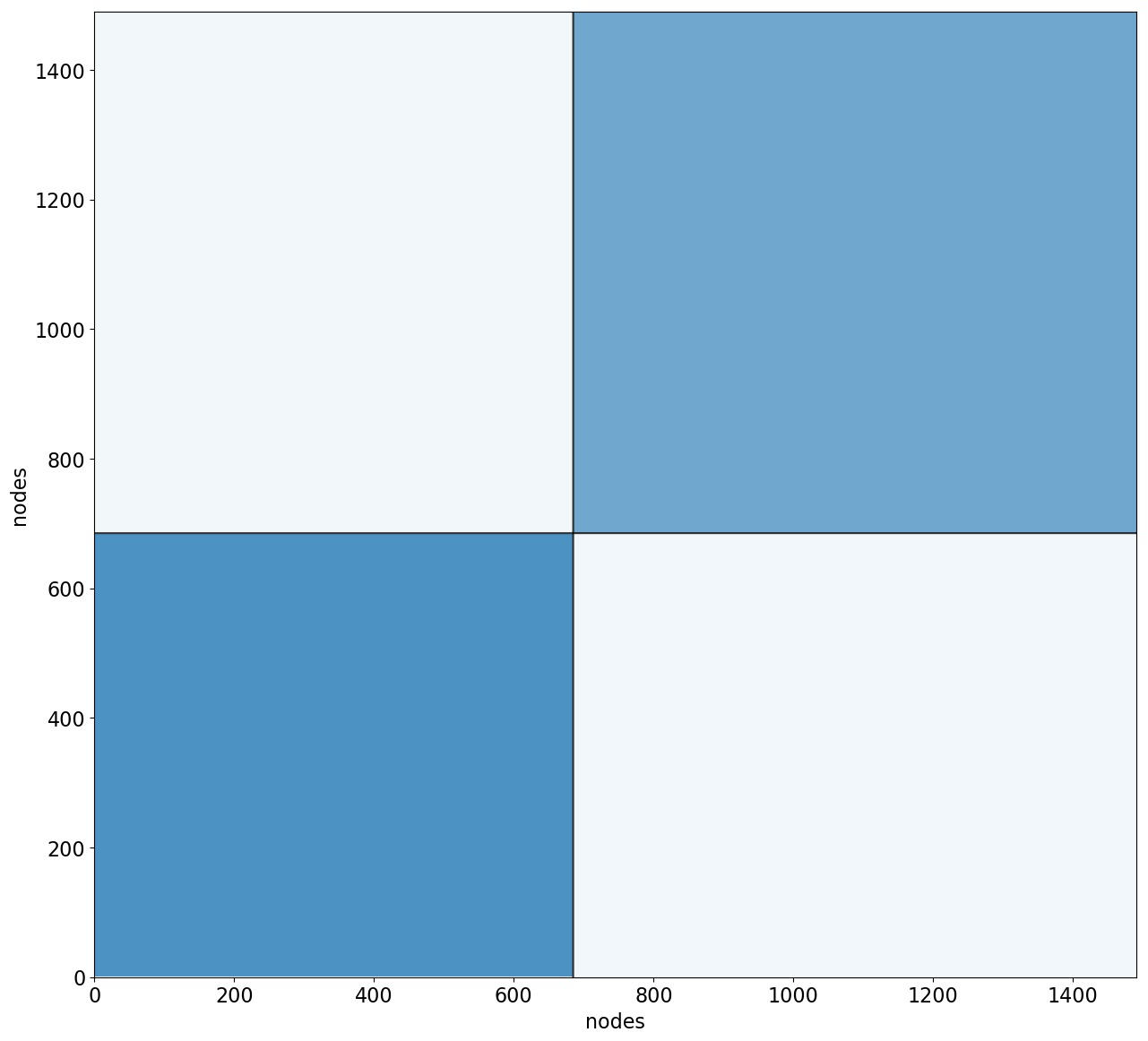}
  \caption{Block densities of reordered adjacency matrix}
\end{subfigure}%
\begin{subfigure}{.5\textwidth}
  \centering
  \includegraphics[width=.7\linewidth]{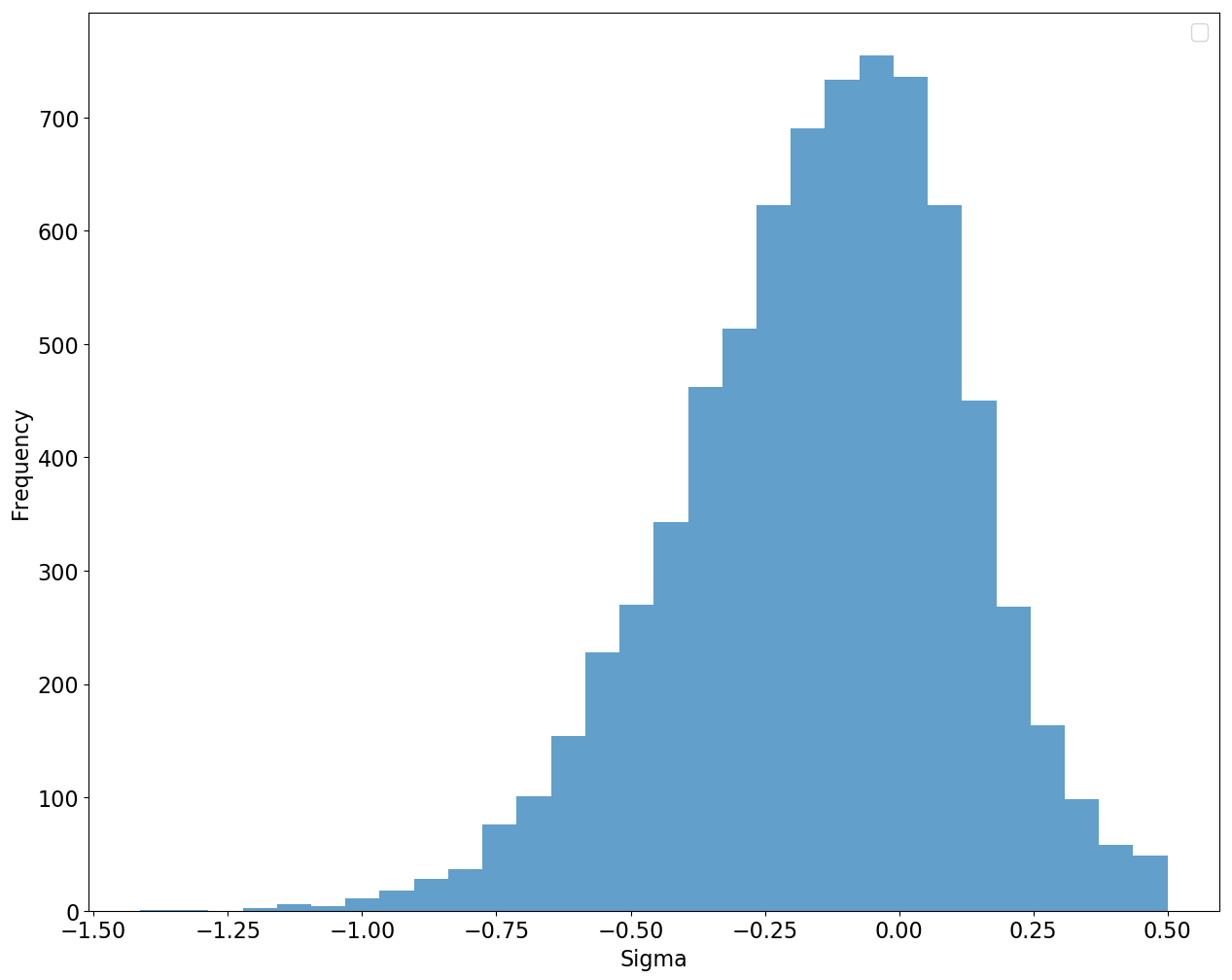}
  \caption{Histogram of $\sigma$}
\end{subfigure}
%\caption{Posterior $K_n$ and $\sigma$ for the Polblogs dataset with $\alpha = 0.8$}
%\end{figure}
%\begin{figure}[H]
\begin{subfigure}{.5\textwidth}
  \centering
  \includegraphics[width=.6\linewidth]{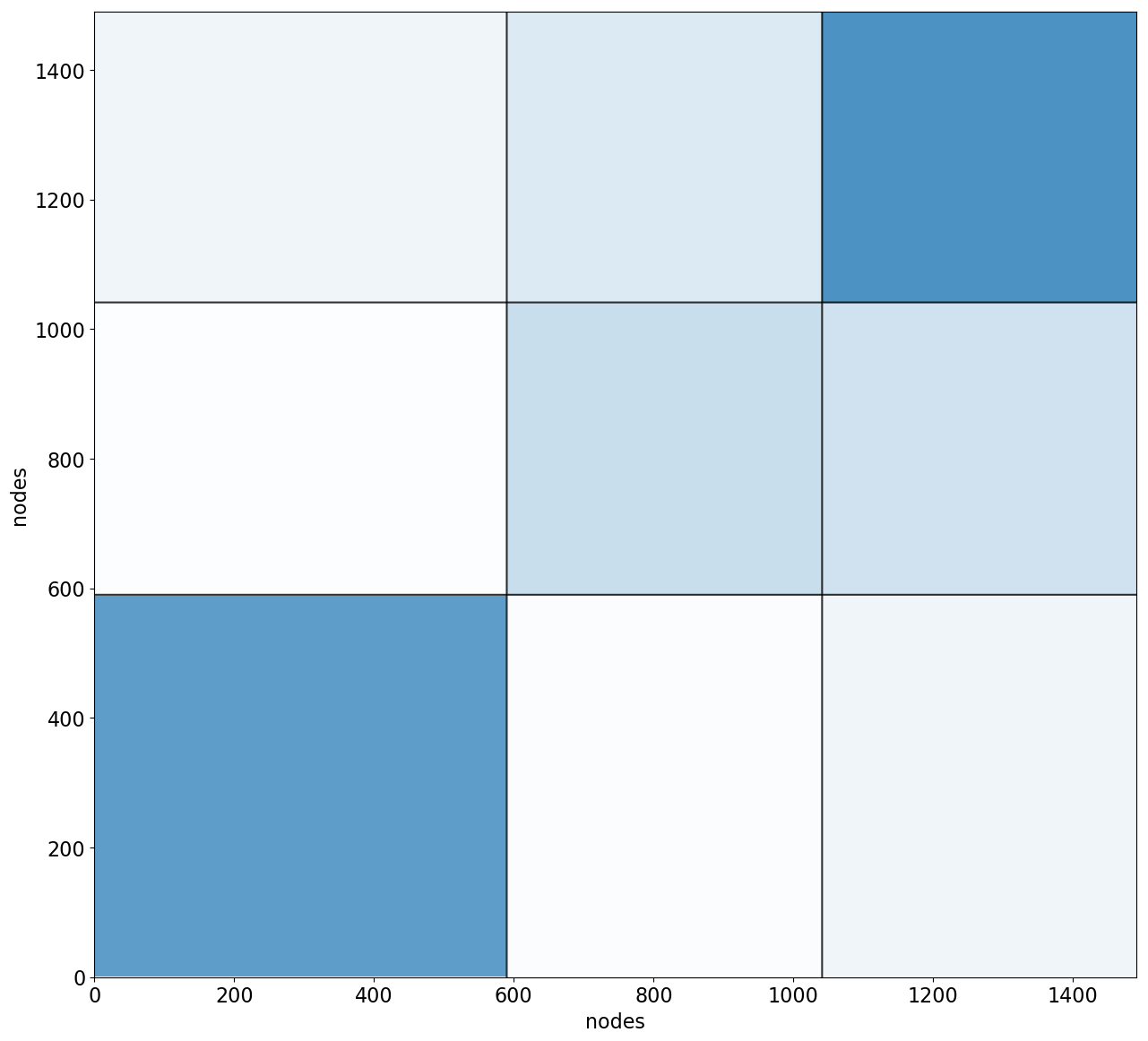}
  \caption{Block densities of reordered adjacency matrix}
\end{subfigure}%
\begin{subfigure}{.5\textwidth}
  \centering
  \includegraphics[width=.7\linewidth]{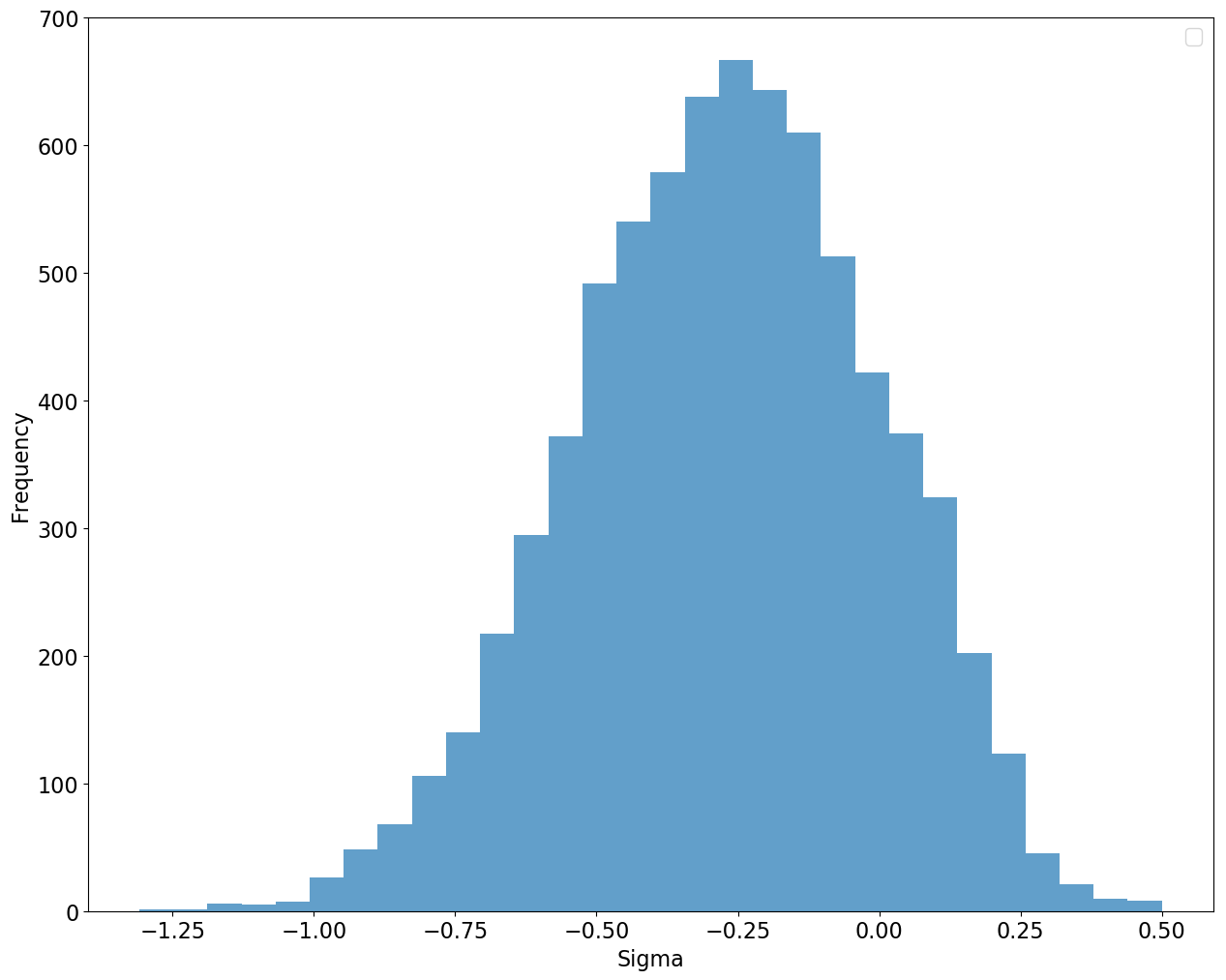}
  \caption{Histogram of $\sigma$}
\end{subfigure}
%\caption{Posterior $K_n$ and $\sigma$ for the Polblogs dataset with $\alpha = 0.4$}
%\end{figure}
%\begin{figure}[H]
\begin{subfigure}{.5\textwidth}
  \centering
  \includegraphics[width=.6\linewidth]{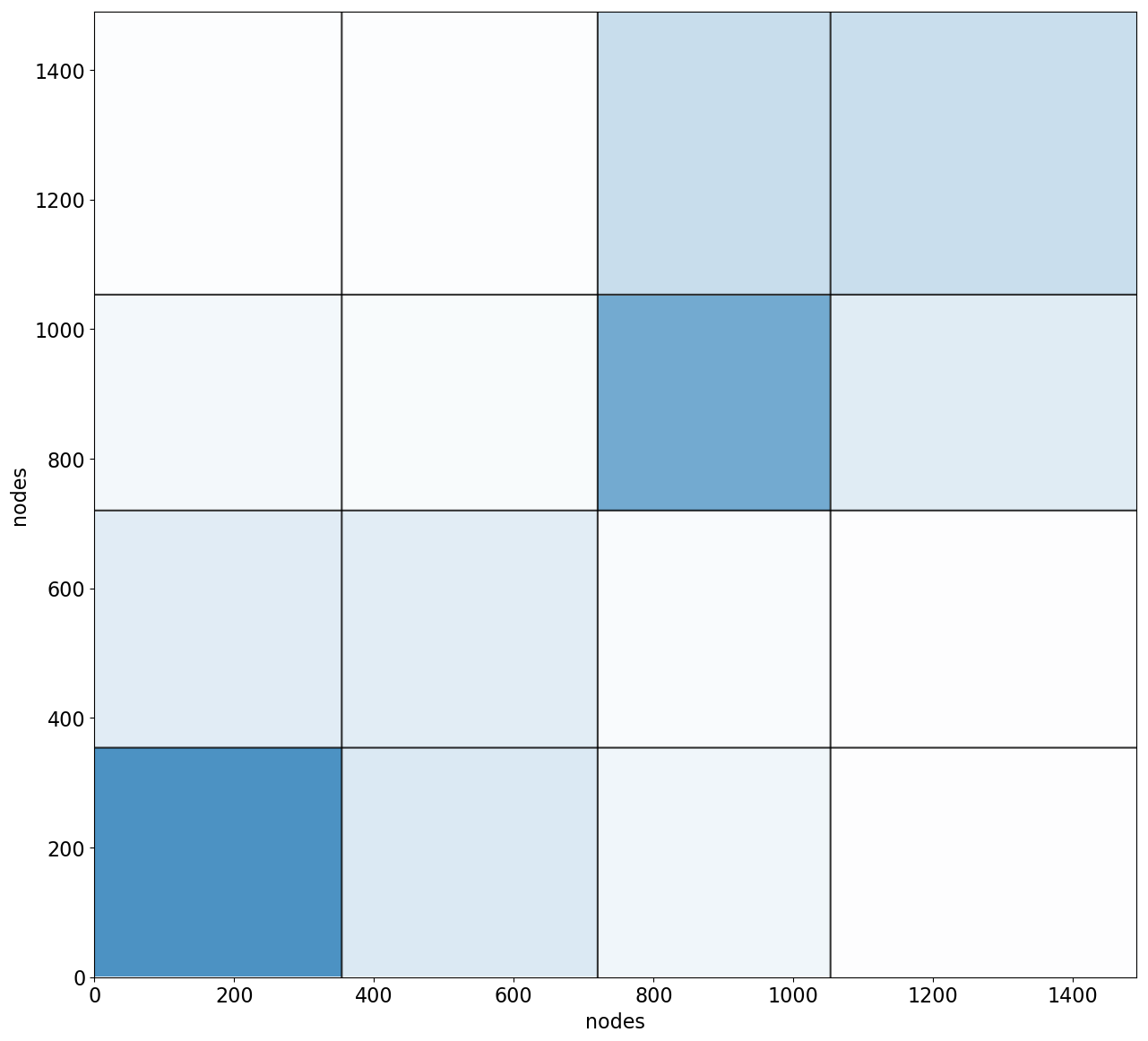}
  \caption{Block densities of reordered adjacency matrix}
\end{subfigure}%
\begin{subfigure}{.5\textwidth}
  \centering
  \includegraphics[width=.7\linewidth]{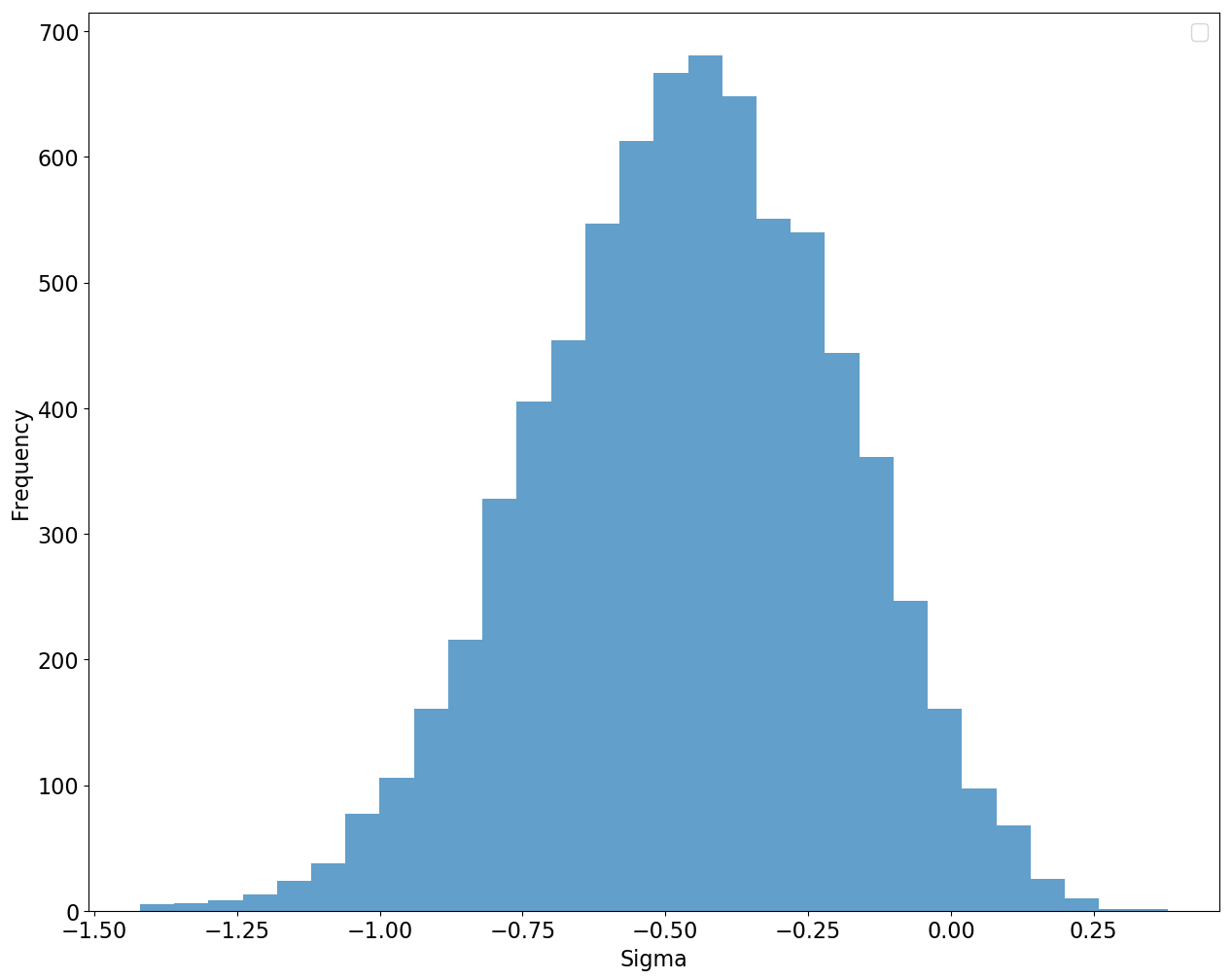}
  \caption{Histogram of $\sigma$}
\end{subfigure}
\caption{(Left) Estimated communities and (right) posterior on $\sigma$ for the Polblogs dataset with (top row) $\alpha = 0.8$, (middle row) $\alpha=0.4$ and (bottom row) $\alpha = 0.2$.}
\label{fig:polblogs}
\end{figure}

\begin{table}[t]
\caption{Proportion of the interactions of the features in each block for different values of overlapping}
\begin{tabular}{|c|cc|ccc|cccc|} \hline
\text{} & \multicolumn{2}{c|}{$\alpha=0.8$} & \multicolumn{3}{c|}{$\alpha=0.4$} & \multicolumn{4}{c|}{$\alpha=0.2$} \\ \hline\hline
			 & 1 & 2 & 1 & 2 & 3 & 1 & 2 & 3 & 4 \\ \hline
			\text{Dem/Dem} & 91.5 & 0.7 & 93.2 & 1.1  & 0.8 & 95 & 84   & 1.5  & 0.1\\
			\text{Dem/Rep } & 8.0 & 9.5 & 6.5  & 13   & 5.1 & 5  & 14.5 & 15.5 & 4.8\\
			\text{Rep/Rep} & 0.5 & 89.8 & 0.3  & 85.9 & 94.1& 0  & 1.5  & 83   & 95.1\\ \hline

\end{tabular}

\label{tab:polblogs}
\end{table}

\subsection{Wikipedia topcast}
The network is a partial web graph of Wikipedia hyperlinks collected in September 2011~\citep{klymko2014using}. It is a directed unweighted graph where an edge $(i,j)$ corresponds to a citation from a page $i$ to page $j$. We restrict it to the first $3000$ nodes, and the associated $5687$ edges.  We run three MCMC chains for  $200,000$ iterations. Trace plots of the number of active communities and parameter $\sigma$ are given in Figure \ref{fig:wiki}. Figure~\ref{fig:wiki_edge_partition} shows the adjacency matrix reordered by communities, as explained in the previous section. In order to check that the learnt communities/features are meaningful, we report in Figures the proportions of webpages associated to a given category within a given community/feature (note that a webpage can be associated to multiple categories hence the proportion do not sum to 1).

\begin{figure}[H]
\begin{subfigure}{.5\textwidth}
  \centering
  \includegraphics[width=.7\linewidth]{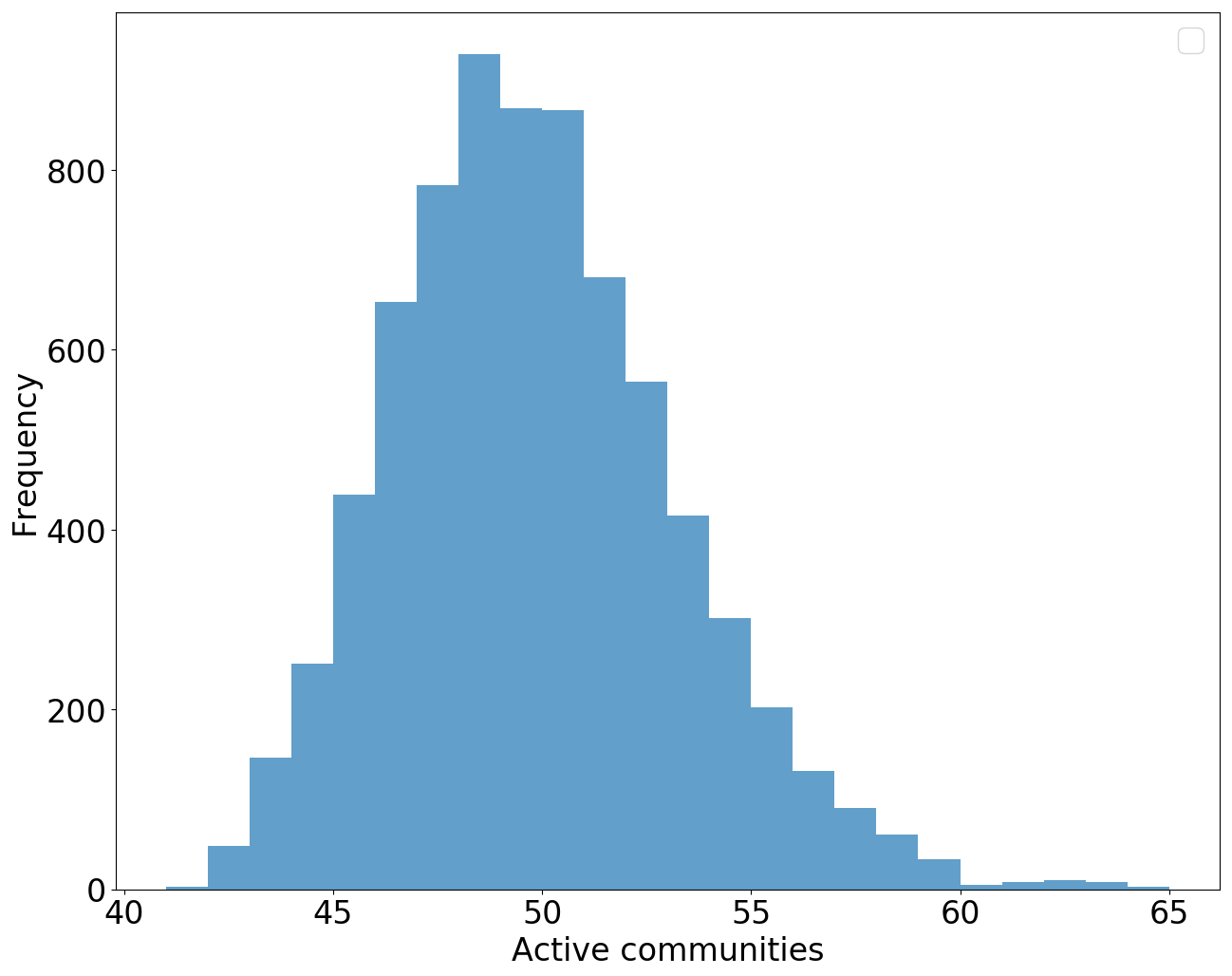}
  \caption{Histogram of number of communities}
  \label{fig:wiki_active_co}
\end{subfigure}%
\begin{subfigure}{.5\textwidth}
  \centering
  \includegraphics[width=.7\linewidth]{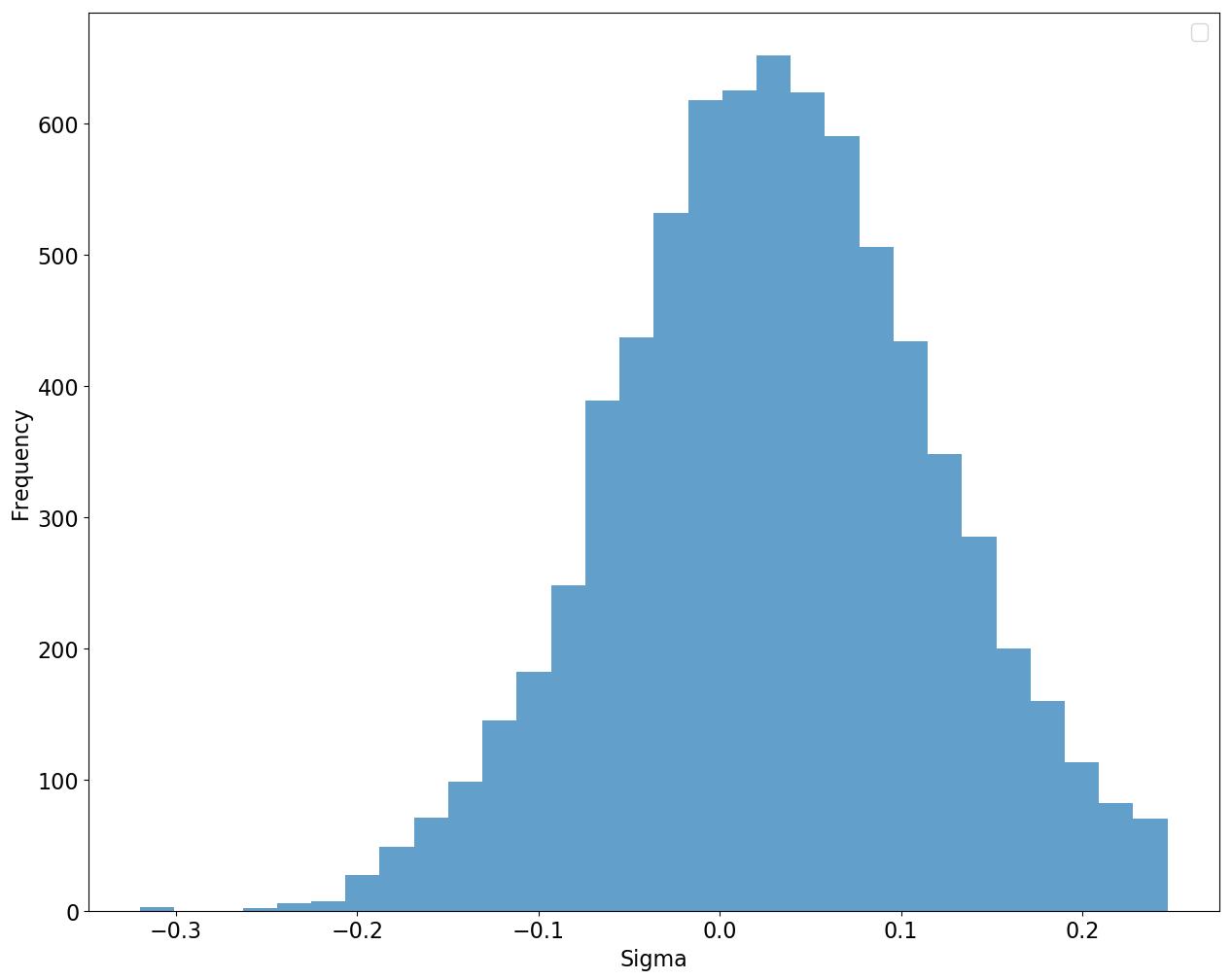}
  \caption{Histogram of $\sigma$}
  \label{fig:wiki_sigma}
\end{subfigure}
\caption{Posterior of $K_n$ and $\sigma$ for the Wiki-topcats dataset}
\label{fig:wiki}
\end{figure}
\begin{figure}[H]
  \centering
  \includegraphics[width=.4\linewidth]{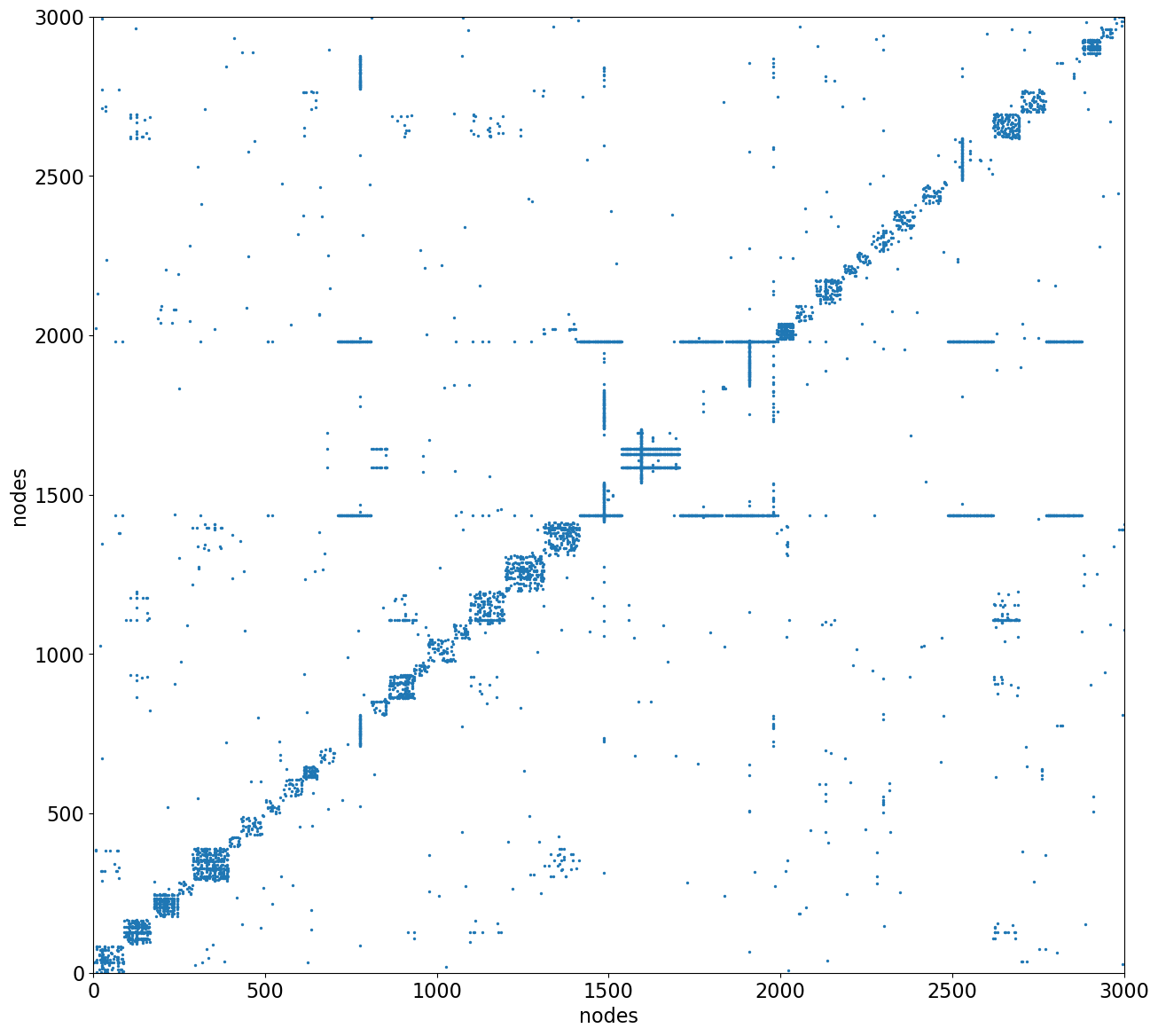}
  \caption{Reordered adjacency matrix of the Wikipedia topcats dataset}
  \label{fig:wiki_edge_partition}
\end{figure}

Note that, while the approach is able to estimate the latent block-structure, this dataset has the particularity of having star nodes, a feature that is not captured by our model.

\begin{figure}[H]
     \includegraphics[width=.7\linewidth]{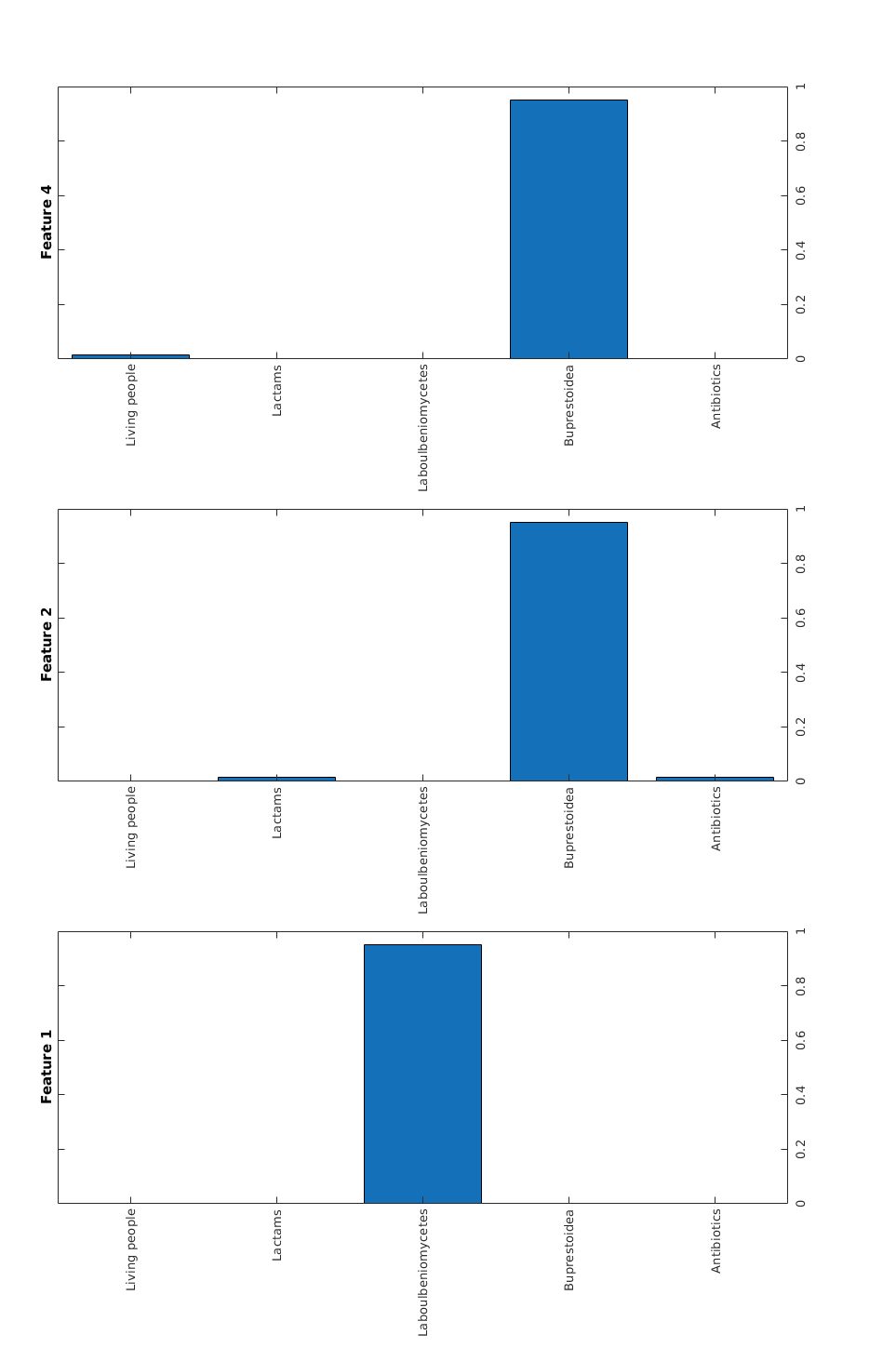}
\caption{Features compared to categories for the wikipedia dataset}
\label{fig:wikicat1}
\end{figure}

\begin{figure}[H]
     \includegraphics[width=.9\linewidth]{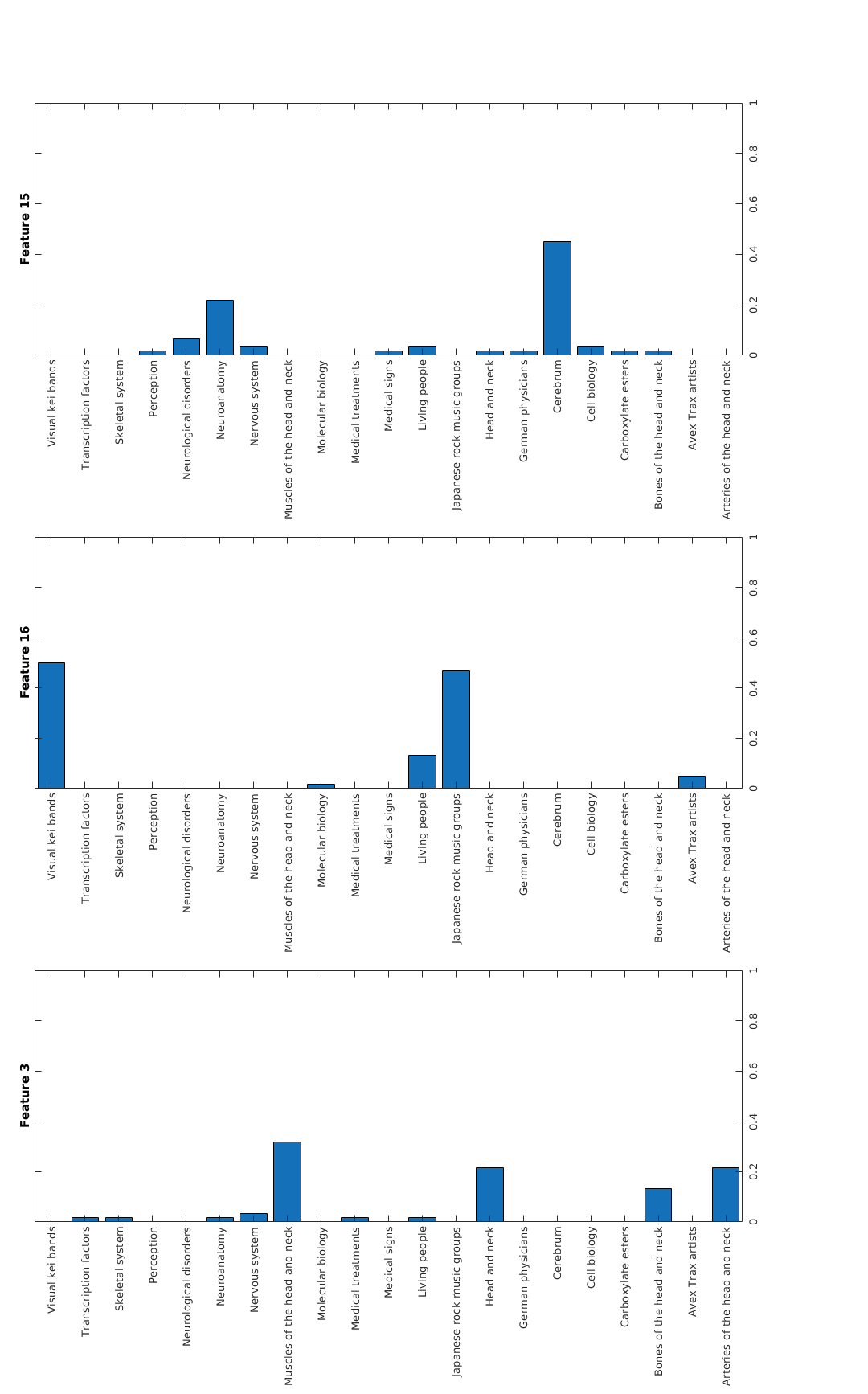}
\caption{Features compared to categories for the wikipedia dataset}
\label{fig:wikicat2}
\end{figure}

\subsection{Deezer}
The dataset was collected from the music streaming service Deezer in November 2017~\citep{Rozemberczki2018}. It represents the friendship network of a subset of Deezer users from Romania. It is an undirected unweighted graph where nodes represent the users and edges are the mutual friendships. There are $41 773$ nodes and $125 826$ edges. We run three chains with $100 000$ iterations each. Posterior histograms of the number of active communities and $\sigma$ are given in Figure~\ref{fig:deezer}. The algorithms finds around $45$ communities/features for this dataset. The reordered adjacency matrix and block densities based on the point estimate of the partition are given in Figure~\ref{fig:deezer_clust}.

\begin{figure}[H]
\begin{subfigure}{.5\textwidth}
  \centering
  \includegraphics[width=.8\linewidth]{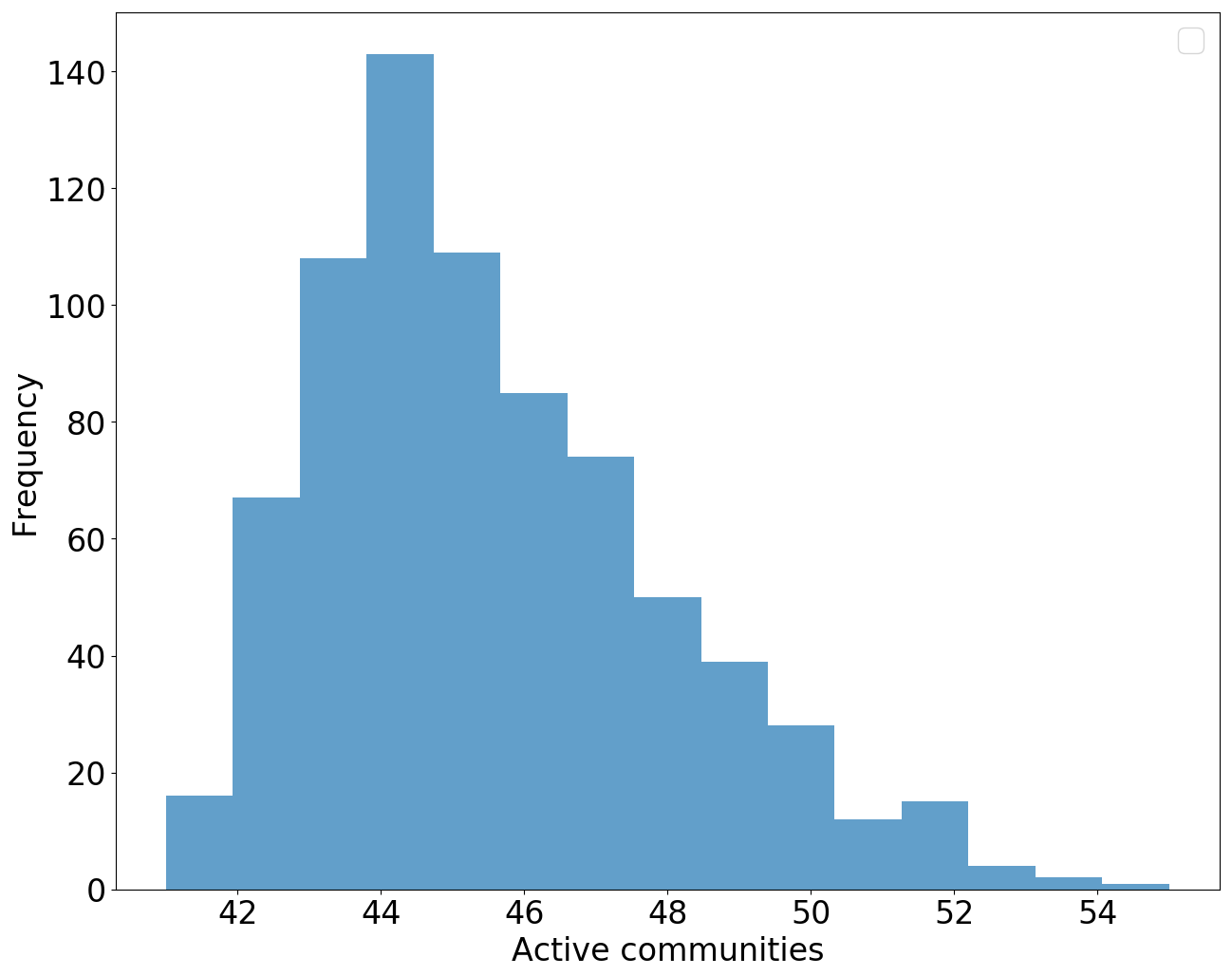}
  \caption{Histogram of number of communities}
  \label{fig:deezer_active_co}
\end{subfigure}%
\begin{subfigure}{.5\textwidth}
  \centering
  \includegraphics[width=.8\linewidth]{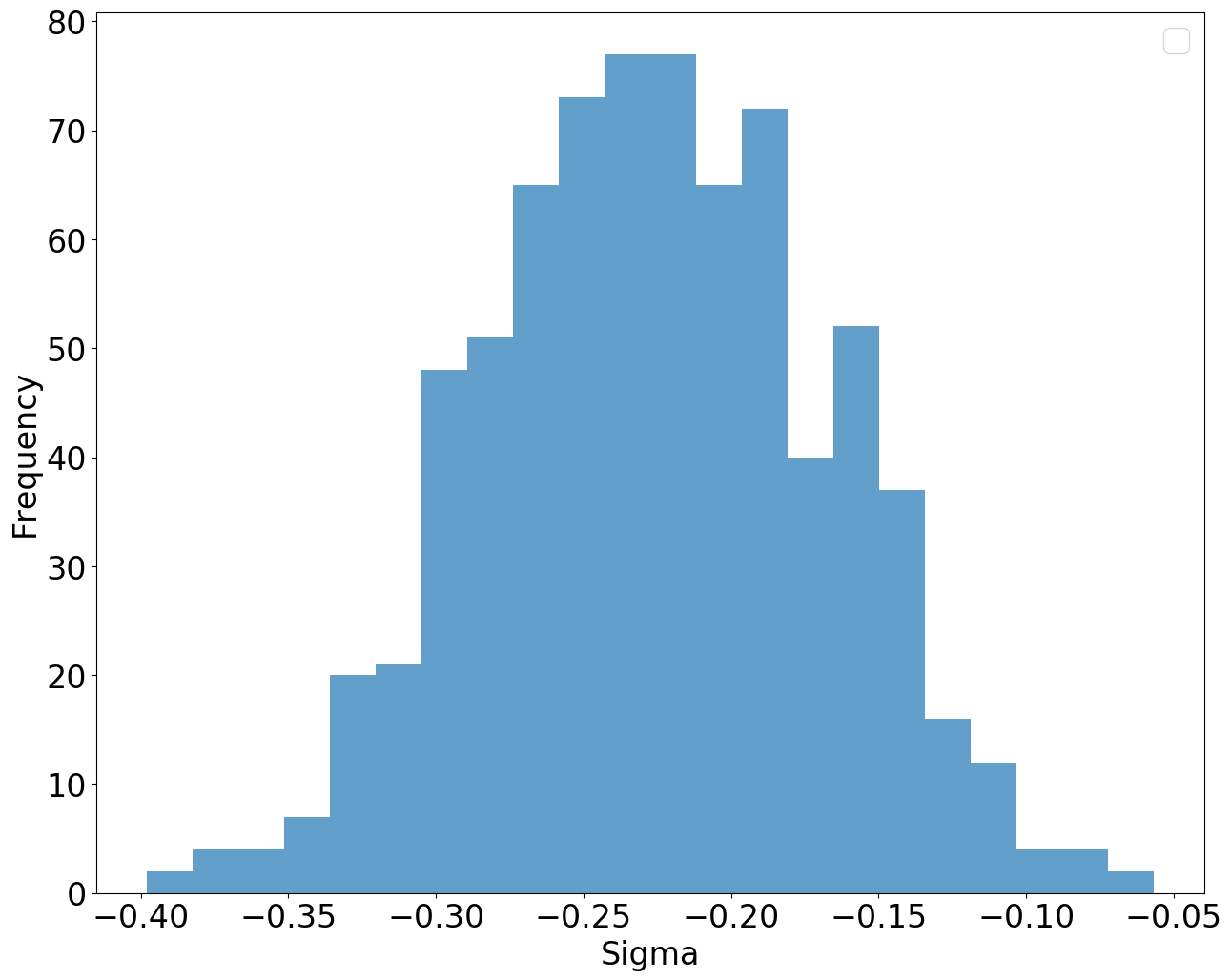}
  \caption{Histogram of $\sigma$}
  \label{synt:deezer_sigma}
\end{subfigure}
\caption{Posterior of $K_n$ and $\sigma$ on Deezer's dataset}
\label{fig:deezer}
\end{figure}

Now we can reorder the nodes using approximate MAP clustering as previously. We obtain the following adjacency matrix
\begin{figure}[H]
\begin{subfigure}{.5\textwidth}
  \centering
  \includegraphics[width=.8\linewidth]{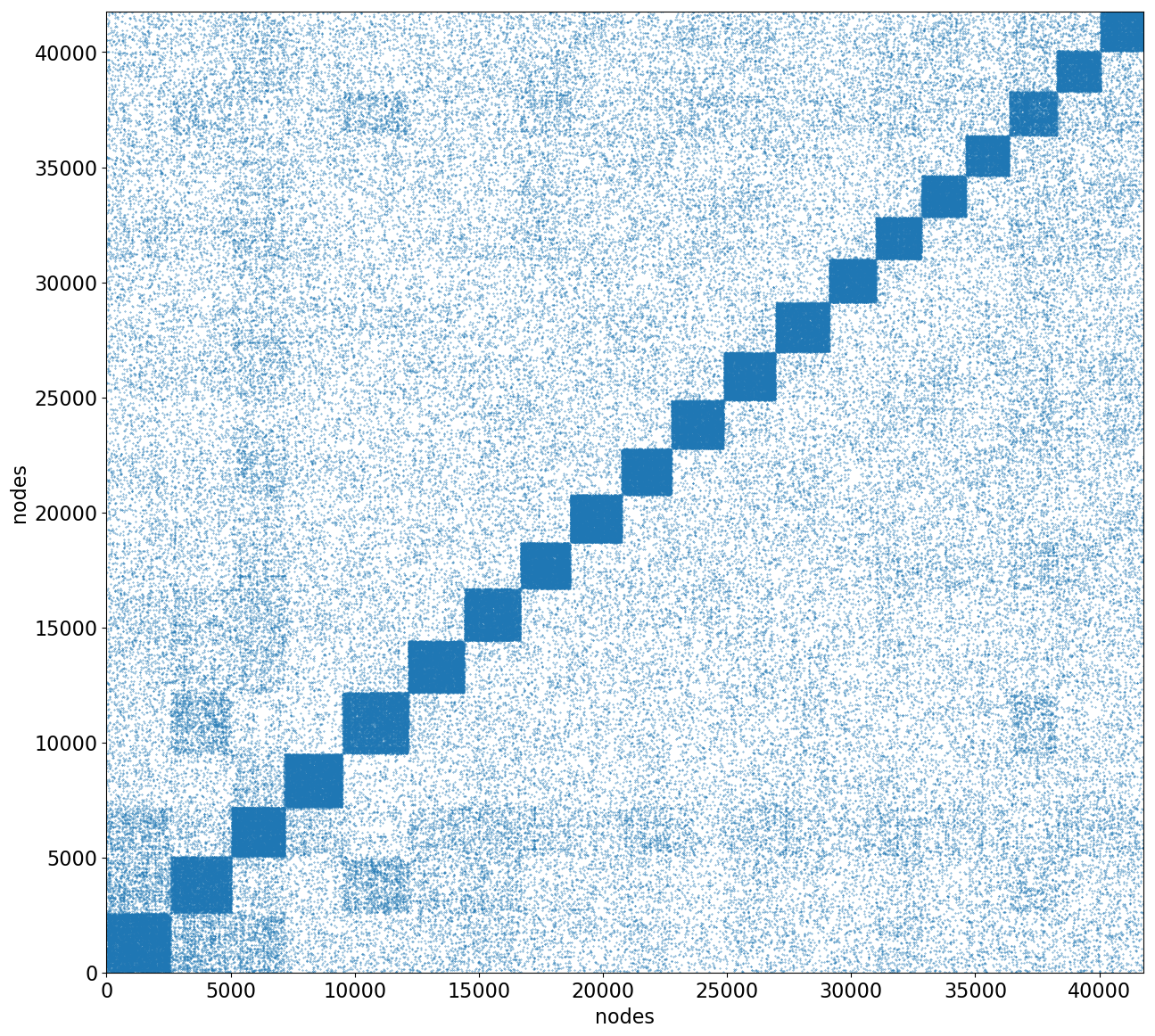}
  \caption{Reordered Adjacency matrix}
  \label{fig:deezer_adj}
\end{subfigure}%
\begin{subfigure}{.5\textwidth}
  \centering
  \includegraphics[width=.8\linewidth]{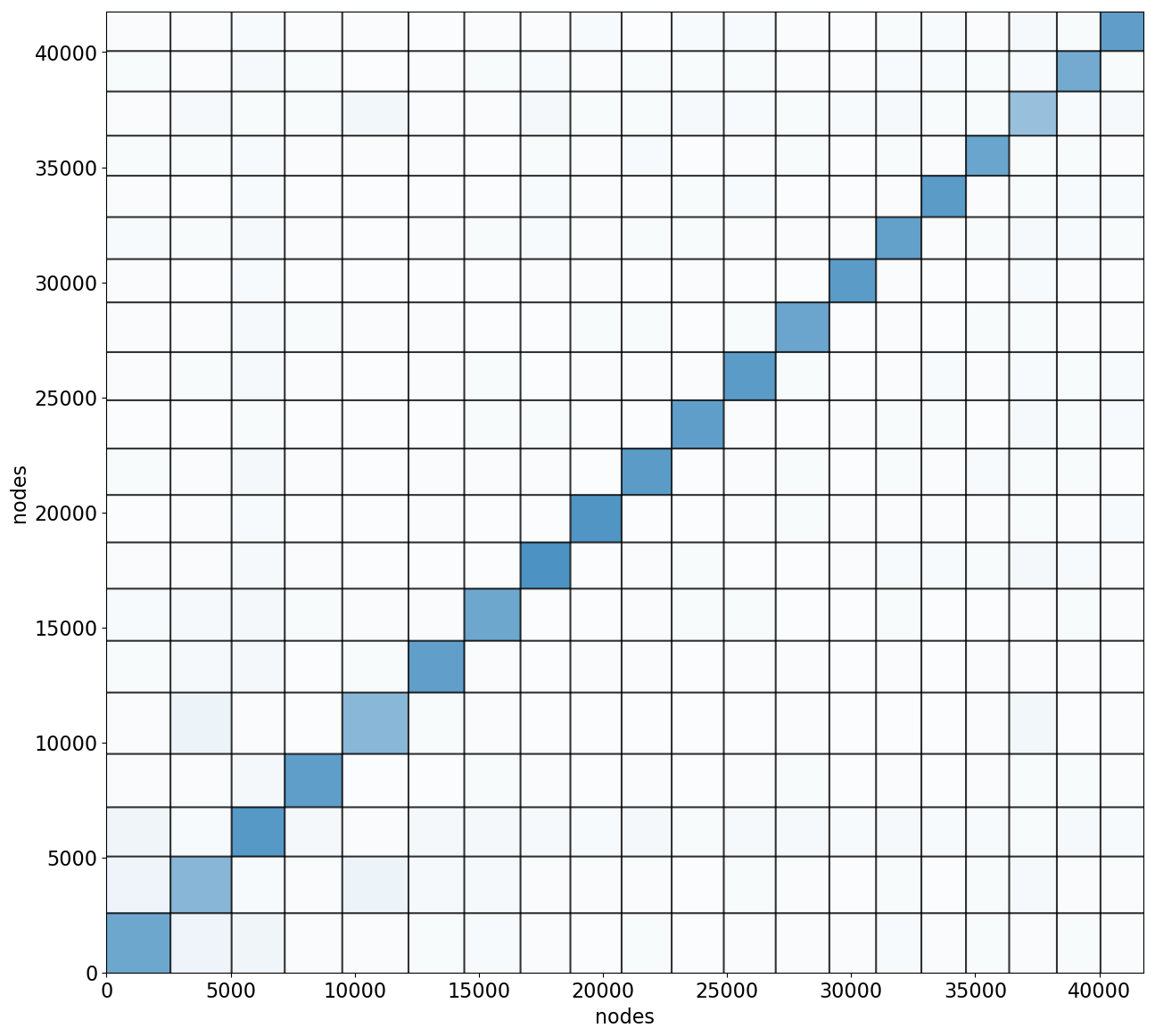}
  \caption{Block densities of reordered adjacency matrix}
  \label{synt:deezer_adj_den}
\end{subfigure}
\caption{Reordered adjacency matrix and block densities for Deezer's dataset.}
  \label{fig:deezer_clust}
\end{figure}

For each individual in the network, a list of musical genres liked by that person are available. There are in total 84 distinct genres. We represent in Figure~\ref{fig:deezergenres} the proportion of individuals who liked a subset of the 84 genres for three different communities where the interpretation in terms of genres is quite clear. The overall proportion of individuals liking a given genre is shown at the bottom of Figure~\ref{fig:deezergenres}. If the bar is red, this indicates that the proportion is 10\% higher in the community than in the population. If the bar is blue, this means the proportion is 10\% lower. Community 11 can be interpreted as $R\&B$, Community 8 as Dance, and Community 3 as Rock music. For some of the communities, not reported here, the interpretation in terms of the liked genres is less clear, and may be due to other covariates. 

\begin{figure}[H]
     \includegraphics[width=.9\linewidth]{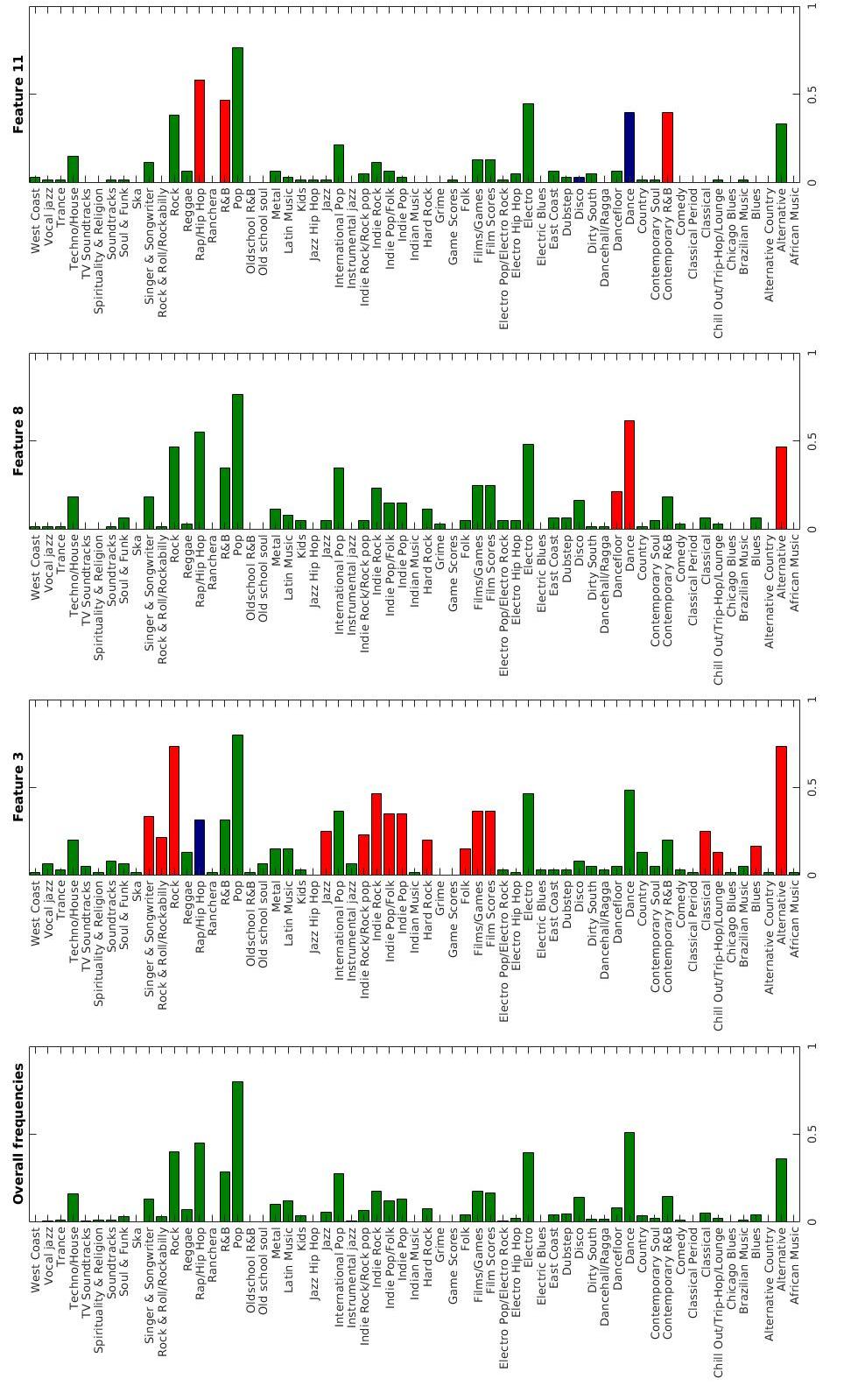}
\caption{Features compared to genres for Deezer's dataset.}
\label{fig:deezergenres}
\end{figure}

\section{Discussion}

The model presented in this paper assumed the same parameter $\beta$ for each node. We can also consider a degree corrected version of the model, similarly to \citet{Zhou2015}, where each node is assigned a different parameter $\beta_i > 0$ and then defining $Z_{ijk} \sim \Poisson(\frac{r_k v_{ik} v_{jk}}{\beta_i \beta_j})$. It is unclear however if a MCMC sampler targeting the exact posterior distribution could be implemented, and one may need to resort to some truncation approximation as in \citet{Zhou2015}.

The count matrix $(A_{ij})$ is infinitely exchangeable, hence the model presented in this article lead to asymptotically dense graphs. That is, $\sum_{1\leq i,j\leq n} A_{ij}\asymp n^2$ as $n$ tends to infinity. In order to obtain sparse graphs, we could consider two different strategies. The first solution consists in dropping the infinite exchangeability property and take $\beta^{(n)}_i \rightarrow \infty$ with $n$, then the number of edges will behave as $(n/\beta^{(n)})^2$ (we can for instance take $\beta^{(n)}_i = \sqrt n$ for any node $i$ to obtain a linear growth of the number of edges). The  model would still be finitely exchangeable for any fixed $n$, but not projective anymore. The second solution would be to consider the different notion of infinite exchangeability developed in \citep{Caron2017} and consider $(\beta_i)_i$ as a realization of a Poisson point process.

Finally, we presented a model for count (and binary) data. The results build on the additive contributions of the communities, which is why we chose the Poisson distribution on the entries of the adjacency matrix $(A_{ij})$. We can generalize to non count data using other probability distributions which are closed under convolution. For example, one could consider $A_{ij} \sim \Gam(\sum_k r_k v_{ik} v_{jk},1)$ for $A_{ij} \in \mathbb{R}_+$ or $A_{ij} \sim \mathcal{N}(\sum_k r_k v_{ik} v_{jk},1)$ for $A_{ij} \in \mathbb{R}$.

\bibliography{bibli}
\bibliographystyle{plainnat}

\appendix

%\fc{add somewhere details of the Laplace exponent and $\kappa$ for the GGP}

\section{Proofs}
\label{sec:proofs}
\subsection{Technical Lemmas}

\begin{lemma}\label{tail}
(\cite{gnedin2007notes}, Propositions 17 and 19) Let $\rho$ be a L\'evy measure, let $\overline{\rho}(x) = \int_x^\infty \rho(r) dr$ be the tail Levy intensity and $\psi (t) = \int (1-e^{- r t}) \rho(dr)$ its Laplace exponent. Then the two following conditions are equivalent:
\begin{eqnarray}
\overline{\rho}(x) &\stackrel{x \rightarrow 0_+}{\sim}& \ell(1/x) x^{- \sigma}\label{condition:tail} \\
\psi(t) &\stackrel{x \rightarrow +\infty}{\sim}& \Gamma (1-\sigma)  t^{\sigma} \ell(t)
\end{eqnarray}
with $\ell$ a slowly varying function and $0 \leq \sigma < 1$.\\

Besides, if we let $\psi_d(t) = \frac{t^d}{d!} \int r^d e^{-rt} \rho(dr)$
\begin{enumerate}
\item if $\sigma > 0$, then $(\ref{condition:tail})$ implies that $\psi_d(t) \stackrel{t \rightarrow +\infty}{\sim} \frac{\sigma \Gamma (d-\sigma)}{d!} t^{\sigma} \ell(t)$
\item if $\sigma = 0$, then $(\ref{condition:tail})$ implies that $\psi_d(t) = o(\ell(t))$
\end{enumerate}
\end{lemma}

\begin{lemma}\label{poisson:concentration}\cite[Exercise 15]{Pollard2015}
Let $X$ be a Poisson random variable with parameter $\lambda $ . For any $t>0$
\begin{equation}
\mathbb{P}(|X-\lambda| \geq \lambda t) \leq 2 e^{-\frac{\lambda t^2}{2(1+t)}}.
\end{equation}
\end{lemma}

\begin{lemma} \label{poisson:as}
Let $(X_n)_{n\geq 1}$ be a sequence of Poisson random variables with mean $(\mu_n)_{n\geq 1}$.  If $\log n = o(\mu_n)$ then $X_n \sim \mu_n$ almost surely as $n$ tends to infinity.
%\begin{enumerate}
%\item Suppose that $\log n = o(\mu_n)$, then $X_n \sim \mu_n\ \ \ a.s.$
%\item Suppose that $\log n = O(\mu_n)$, then $X_n = O(\mu_n)\ \ \ a.s.$
%\end{enumerate}
\end{lemma}

\begin{proof}

Let $ 0< \epsilon < 1/2$. Using Lemma \ref{poisson:concentration}, we have
\begin{align}
\mathbb{P} \left ( \left |\frac{X_n}{\mu_n} -1 \right | \geq \epsilon \right ) & \leq 2 e^{-\frac{\epsilon^2 \mu_n}{4}}\nonumber\\
&=2 n^{-\frac{\epsilon^2 \mu_n}{4 \log n}}\label{ineq:concentrationX}
\end{align}
%Therefore
%\begin{eqnarray}
%\mathbb{P} [\ |\frac{X_n}{\mu_n} -1 | \geq \epsilon \ ]  \leq 2 n^{-\frac{\epsilon^2 \mu_n}{4 (\log n)}} \label{ineq:concentrationX}
%\end{eqnarray}
Using the assumption, we have that $-\frac{\epsilon^2  \mu_n}{4 \log n} \rightarrow -\infty$. Therefore, the RHS of $(\ref{ineq:concentrationX})$ is summable. The almost sure result follows from Borel-Cantelli lemma.
\end{proof}

\begin{lemma}
For any $x,y\geq 0$, we have the following bound
\begin{align}
1-e^{-xy}\leq \max(1,y)(1-e^{-x})\label{eq:bound}
\end{align}
\end{lemma}

\begin{proof}
The bound is trivial when $y\leq 1$. Consider the case $y>1$. For all $x$, the function $y\rightarrow e^{-xy}$ is convex hence $f_x(y)=\frac{e^{-xy}-1}{y}$ is a monotonically non-decreasing function of $y$ therefore
\begin{align*}
\frac{e^{-xy}-1}{y}\geq e^{-x}-1
\end{align*}
for $y\geq 1$.
\end{proof}
%
%
%\begin{proof}
%
%Suppose first that $\log n = o(\mu_n)$ and let $ 0< \epsilon < 1/2$. Using Lemma \ref{poisson:concentration}, we have
%\begin{eqnarray*}
%\mathbb{P} [\ |\frac{X_n}{\mu_n} -1 | \geq \epsilon \ ]  \leq 2 e^{-\frac{\epsilon^2 \mu_n}{4}}
%\end{eqnarray*}
%Therefore
%\begin{eqnarray}
%\mathbb{P} [\ |\frac{X_n}{\mu_n} -1 | \geq \epsilon \ ]  \leq 2 (\frac{1}{n})^{-\frac{\epsilon^2 \mu_n}{4 (\log n)}} \label{ineq:concentrationX}
%\end{eqnarray}
%Using the assumption, we have that $-\frac{\epsilon^2  \mu_n}{4 (\log n)} \rightarrow -\infty$. Therefore, the RHS of $(\ref{ineq:concentrationX})$ is summable.
%Finally, applying Borell-Cantelli's lemma, we get the a.s asymptotic behaviour. \\
%
%Now suppose that $\log n = O(\mu_n)$. Therefore, we can take $m,M > 0$ such that for $n$ large enough
%$$ m \log n \leq \mu_n$$
%Let $C = \max(e^2, 2/m)$, using Lemma \ref{poisson:concentration} we find
%$$ \mathbb{P} [\ \frac{X_n}{\mu_n} -1  \geq C \ ] \leq e^{-\mu_n (\log C - 1)C} \leq \frac{1}{n^2} $$
%Therefore we conclude using Borel-Cantelli. \\
%
%
%\end{proof}

\subsection{Proofs of Section 3}

\begin{proof}[Proof of Proposition \ref{lemma:Kpoisson}]

The result for $K_n$ is proved in \cite{james2014poisson} in the general context of GIBP. We provide here the details of the proof for $K_n$, which can be straightforwardly adapted to $K_{n,j}$.

First, let us remark that the bound \eqref{eq:bound} together with assumptions \eqref{eq:assumpt1} and \eqref{eq:assumpt2} imply $\Psi(n)<\infty$. For $s < 0$,
\begin{eqnarray*}
\mathbb{E} [e^{s K_n}] &=& \mathbb{E} [\ \prod\limits_{k=1}^\infty \mathbb{E} [e^{s \mathbb{1}_{\sum\limits_{1\leq i,j \leq n} Z_{ijk} \geq 1}} | G]\  ] \\
&=& \mathbb{E} [\ \prod\limits_{k=1}^\infty [\ e^s + (1-e^s) e^{-r_k (\sum\limits_{i=1}^n v_{ik})^2} \ ] \\
&=& \mathbb{E} [\ e^{\sum_k \log [e^s + (1-e^s) e^{-r_k (\sum\limits_{i=1}^n v_{ik})^2}]} \ ]
\end{eqnarray*}
Then, since
\begin{eqnarray*}
\log [e^s + (1-e^s) e^{-r_k (\sum\limits_{i=1}^n v_{ik})^2}] &=& \log [e^{-r_k (\sum\limits_{i=1}^n v_{ik})^2} + e^s (1- e^{-r_k (\sum\limits_{i=1}^n v_{ik})^2})] \\
&\leq & \log [1 + e^s (1- e^{-r_k (\sum\limits_{i=1}^n v_{ik})^2})] \\
&\leq& 1 - e^{-r_k (\sum\limits_{i=1}^n v_{ik})^2}
\end{eqnarray*}
and the last part is integrable, we can use Campbell's theorem~\citep{Kingman1993} to get:
\begin{eqnarray*}
\mathbb{E} [e^{s K_n}] &=& \exp [\ \int (e^{\log [e^s+ (1-e^s)e^{-r (\sum\limits_{i=1}^n v_{i})^2}] } -1 ) \prod\limits_{i=1}^n F(dv_i)\ \rho(dr)  \ ] \\
&=& \exp [\ (e^s-1) \Psi(n) \ ]
\end{eqnarray*}
We can prove similarly that $K_{n,d}$ is a Poisson random variable with mean $\Psi_d(n)$ and that $\sum\limits_{d \geq D} K_{n,d}$ is Poisson distributed with mean $\sum\limits_{d \geq D} \Psi_d(n)$. The assumption $\Psi_d(n) < \infty$ is also sufficient in this case to apply Campbell's theorem.

\end{proof}

%\begin{proof}[Proof of Proposition \ref{lemma:facto}]
%The first point is a straight forward application of the law of large numbers.
%
%Let $\nu = \sum_k \delta_{r_k}$. Then, for all $x>0$, $\nu([x,+\infty [ )$ is distributed as $Poisson(\overline{\rho}(x)) $. Let us show that,
%$$\nu([1/x,+\infty]) \stackrel{x\rightarrow +\infty}{\sim} x^{\sigma} \ell(x)\ \ \ a.s.$$
%Using Lemma \ref{poisson:as}{poisson:as} on the sequence $(\nu([1/k,+\infty]))_{k \geq 1}$, we find that
%$$\nu([1/k,+\infty]) \stackrel{k\rightarrow +\infty}{\sim} k^{\sigma} \ell(k)\ \ \ a.s.$$
%Now, since $x \mapsto \nu([1/x,+\infty])$ is almost surely non decreasing, it comes that
%$$\nu([1/\lfloor x \rfloor,+\infty]) \leq \nu([1/x,+\infty]) \leq \nu([1/\lfloor x+1 \rfloor,+\infty])$$
%We get the desired result by noticing that
%$$(\lfloor x+1 \rfloor)^\sigma \ell (\lfloor x+1 \rfloor) \sim (\lfloor x \rfloor)^\sigma \ell (\lfloor x \rfloor) \sim  x^\sigma \ell (x) $$
%Now, using Proposition 23 of \cite{gnedin2007notes}, it comes that
%$$ r_{[k]} \sim k^{-1/\sigma} \ell'(k)$$
%where $\ell'$ is a slowly varying function. We conclude using the law of large number to get the almost sure asymptotic behavior of $||v_{[k]}||_1$
%\end{proof}

\begin{proof}[Proof of Proposition \ref{lemma:asymptotic_K}]
From Proposition \ref{lemma:Kpoisson}, we get that
\begin{eqnarray}
\mathbb{E}(K_n) = \int (1-e^{-r (\sum\limits_{i=1}^n v_i)^2 }) \prod\limits_{i=1}^n F(dv_i) \ \rho(dr)  = \Psi(n) \label{exp:Z}
\end{eqnarray}
Let $(V_i)_{i \in \mathbb{N}}$ be i.i.d random variables with distribution $F$. By assumption, $0 < \mathbb{E}[V_i] = \mu < +\infty$ and $\mathbb{V}ar [V_i] = \tau^2 < +\infty$. Let $\epsilon > 0$. Let $A(r)$ be defined for $r > 0$ by
$$A(r) = \mathbb{E} [ 1 - e^{-r (\sum\limits_{i=1}^n V_i)^2} ]$$
Since $v\mapsto 1 - e^{-r v}$ is concave, using successively Jensen's inequality and the independence of $(V_i)$, we obtain
\begin{eqnarray*}
A(r) &\leq & 1 - e^{-r \mathbb{E}[(\sum\limits_{i=1}^n V_i)^2]}  \\
&\leq & 1 - e^{-r (n^2 \mu^2 + n \tau^2)} \\
&\leq & 1 - e^{- (1+\epsilon) r n^2 \mu^2}
\end{eqnarray*}
where the last inequality holds for any $\epsilon>0$ when $n>\frac{\tau^2}{\epsilon\mu^2}$. Therefore, for $\epsilon>0$ and $n>\frac{\tau^2}{\epsilon\mu^2}$
$$ \Psi(n) \leq \int (1-e^{-(1+\epsilon) r n^2 \mu^2}) \rho(dr).$$

Besides, since $ v\mapsto 1 - e^{-r v} $ is increasing, by Markov's inequality we have for any $\epsilon>0$
\begin{eqnarray*}
A(r) &=& \mathbb{E} \left [ 1 - e^{-r n^2 (\frac{1}{n} \sum\limits_{i=1}^n V_i)^2} \right ] \\
&\geq & \mathbb{P} \left (\frac{1}{n} \sum\limits_{i=1}^n V_i \geq \frac{\mu}{(1+\epsilon)}\right ) \left (1-e^{-\frac{n^2 \mu^2}{(1+\epsilon)^2} r }\right )
\end{eqnarray*}

Hence
$$\mathbb{P} \left (\frac{1}{n} \sum\limits_{i=1}^n V_i \geq \frac{\mu}{(1+\epsilon)}\right ) \psi\left (\frac{n^2 \mu^2}{(1+\epsilon)^2} \right ) \leq \ \Psi(n)\ \leq \psi\left  ( (1+\epsilon)  n^2 \mu^2\right ) $$
where $\psi (t) = \int (1-e^{- r t}) \rho(dr)$ is the Laplace exponent. Furthermore, by the law of large numbers,
$$\mathbb{P} \left (\frac{1}{n} \sum\limits_{i=1}^n V_i \geq \frac{\mu}{(1+\epsilon)} \right ) \rightarrow 1 .$$
Therefore,  under Assumption \eqref{assumption:reg_var}, Lemma \ref{tail} implies
$$ \mathbb{E}(K_n) \sim \Gamma(1-\sigma)\mu^{2\sigma} n^{2\sigma} \ell(n^2) $$
as $n$ tends to infinity.

In the finite-activity case, that is  $\sigma=0$ and $\ell(t)\rightarrow \overline \rho(0)<\infty$, we have $ \mathbb{E}(K_n) \rightarrow  \overline \rho(0)$ hence $K_n$ tends in distribution to $\Poisson(\overline \rho(0))$.

 Now, for $\sigma > 0$, the almost sure result \eqref{K_asympt} follows from Lemma $\ref{poisson:as}$ and the fact that for every slowly varying function $\ell_0 $ and every $\epsilon>0$
$$\lim_{x \rightarrow \infty} \ell_0(x) x^{-\epsilon} = 0. $$
%Let us also remark that for the particular case of the Gamma process where $\ell(x) \propto \log x$, Lemma \ref{poisson:as} gives that $K_{n} = O(\log n)\ \ \ a.s. $\\

Finally, assume that $(K_n)_{n\geq 1}$ is non-decreasing. We only need to prove the asymptotic behavior for $\sigma = 0$. In that setting, $\Psi(n) \sim \ell(n^2)$. Using the assumption that $\int \rho(dr) = \infty$, we therefore have $\lim_{n\rightarrow\infty} \Psi(n) = \infty$. Let $n \geq 1$,
$$ \Psi(n+1) - \Psi(n) = \int \mathbb{E}\left [ e^{-r (\sum\limits_{i=1}^{n} V_i)^2} - e^{-r (\sum\limits_{i=1}^{n+1} V_i)^2}\right ] \rho(r) dr$$
Since $V_i \geq 0 $ a.s, it comes that $(\Psi(n))_n$ is non-decreasing.
Now extend the sequence $(\Psi(n))_n$ to a non-decreasing and continuous function $\Psi$ on $\mathbb{R}_+$ (by linear interpolation for instance). Let $t > 1$, then
$$1 \leq \frac{\Psi(t+1)}{\Psi(t)} \leq \frac{\Psi(3\lfloor t \rfloor)}{\Psi(\lfloor t \rfloor)} \sim \frac{\ell(9\lfloor t \rfloor^2)}{\ell(\lfloor t \rfloor^2)}\rightarrow 1$$
Hence $\lim \frac{\Psi(t+1)}{\Psi(t)} = 1$

Now, for every integer $m \geq \Psi(0)$, choose $t_m$ such that $\Psi(t_m) = m$. We have that $(t_m)$ is non-decreasing and diverges. Since $\Psi$ is increasing, it comes
$$ \Psi(t_{m+1}-1) \leq \Psi( \lfloor t_{m+1} \rfloor - 1) \leq \Psi(t_{m+1}) = m+1$$
Hence, $\Psi(\lfloor t_{m+1} \rfloor - 1) \sim m$.
Then, using Lemma \ref{poisson:as}, we get that
$$ K_{\lfloor t_{m+1} \rfloor - 1} \sim \Psi( \lfloor t_{m+1} \rfloor - 1) \sim m \ \ \ a.s.$$
Finally, let $n \geq \Psi(0)$, let $m_n = \min \{ m \ |\  n \in \{ \lfloor t_{m} \rfloor ,...,\lfloor t_{m+1} \rfloor - 1 \} \}$,
$$ \frac{K_{\lfloor t_{m_n} \rfloor}}{ \Psi( \lfloor t_{m_n+1} \rfloor -1) } \leq \frac{K_n}{\Psi(n)} \leq \frac{K_{\lfloor t_{m_n+1} \rfloor - 1}}{ \Psi( \lfloor t_{m_n} \rfloor) }$$
Since $t_{m_n} \rightarrow \infty$, both bounds converge to $1$ almost surely, which gives the result.
\end{proof}

\begin{proof}[Proof of Proposition \ref{degree:as}]
As for Proposition \ref{lemma:asymptotic_K}, we only need to show that for $d \geq 1$,
$$\Psi_d(n) \sim \frac{\sigma \Gamma(d-\sigma) }{d!} n^{2 \sigma} \mu^{2 \sigma}\ell(n^2)$$
Therefore the proof is very similar to the one of Proposition \ref{lemma:asymptotic_K}. However, there are some technicalities we need to address since here $v \mapsto \frac{v^d}{d!}e^{-rv}$ is neither convex nor decreasing. Like previously, we will lower bound and upper bound $\Psi_d(n)$ by two quantities that are equivalent to $ \frac{\sigma \Gamma(d-\sigma) }{d!} n^{2 \sigma} \mu^{2 \sigma}\ell(n^2)$.

Let us first introduce some notations. Let $(V_i)_{i \in \mathbb{N}}$ i.i.d variables with distribution $F$. Let $S_n = \sum\limits_{i=1}^n V_i$ and $\varphi_d(r,v)$ defined as
$$ \varphi_d(r,v) = \frac{r^d v^{d} }{d!} e^{-r v}.$$

Now, define $A_d(r,n)$ for $r > 0$ by
$$ A_d(r,n) = \mathbb{E}[\varphi_d(r,S_n^2) ] $$
Suppose that $0 < \sigma < 1$, let $\epsilon > 0$, recalling that $\mathbb{E} V_i = \mu $, define $B_\epsilon = [\frac{n\mu}{\epsilon + 1} , (1+\epsilon)n\mu ]$. Let us notice that the law of large numbers gives us that $\mathbb{P} (S_n \in B_\epsilon) \rightarrow 1$.

\begin{enumerate}
\item Lower bound: We have that
\begin{eqnarray*}
\mathbb{E}[\mathbb{1}_{S_n \in B_\epsilon} \varphi_d(r, S_n^2) ] \leq A_d(r,n)
\end{eqnarray*}
Besides, for all $v\in B_\epsilon$ and all $r>0$
$$
\phi_d(r,v)\geq  \frac{r^d \mu^{2d} n^{2d}}{(1+\epsilon)^{2d} d!} e^{-r n^2 (1+\epsilon)^2 \mu^2}
$$
hence
$$ \mathbb{P} (S_n \in B_\epsilon) \frac{r^d \mu^{2d} n^{2d}}{(1+\epsilon)^{2d} d!} e^{-r n^2 (1+\epsilon)^2 \mu^2}  \leq \mathbb{E}[\mathbb{1}_{S_n \in B_\epsilon} \varphi_d(r,S_n^2) ]. $$
Therefore, using Lemma $\ref{tail}$ and since $\mathbb{P} (S_n \in B_\epsilon) \rightarrow 1$, we have for $n$ large enough
$$ \frac{\sigma \Gamma(d-\sigma) }{d! (1+\epsilon)^{4d + 1 - 2\sigma}} n^{2 \sigma} \mu^{2 \sigma}\ell(n^2) \leq \int A_d(r,n) \rho(dr).$$

\item Upper bound: We have that
\begin{eqnarray*}
A_d(r,n) = \mathbb{E}[\mathbb{1}_{S_n \in B_\epsilon} \varphi_d(r, S_n^2) ] + \mathbb{E}[\mathbb{1}_{S_n \not \in B_\epsilon} \varphi_d(r, S_n^2)]
\end{eqnarray*}
Like previously, since
$$\mathbb{E}[\mathbb{1}_{S_n \in B_\epsilon} \varphi_d(r,n S_n) ] \leq \mathbb{P} (S_n \in B_\epsilon) \frac{r^d (1+\epsilon)^{2d}  \mu^{2d}  n^{2d}}{d!} e^{-r \frac{n^2 \mu^2}{(1+\epsilon)^2}} $$
We find that for $n$ large enough,
$$ \int A_d(r,n) \rho(dr) \leq \frac{\sigma \Gamma(d-\sigma) (1+\epsilon)^{4d + 1 - 2\sigma}}{d!} n^{2 \sigma} \mu^{2 \sigma}\ell(n^2) + \int \mathbb{E}[\mathbb{1}_{S_n \not \in B_\epsilon} \varphi_d(r, S_n^2)] \rho(dr)$$
Therefore, we only need to prove that
$$ \int \mathbb{E}[\mathbb{1}_{S_n \not \in B_\epsilon} \varphi_d(r, S_n^2)] \rho(r)dr = o(n^{2\sigma}\ell(n^2)).$$
In order to do so, we split the integral with respect to $r$ in two parts, an integral over $(0,\frac{1}{n^2})$ and an integral over $(\frac{1}{n^2},\infty)$ and show that both are $o(n^{2\sigma}\ell(n^2))$. Since $\varphi_d(r,v) \leq 1$,
\begin{align*}
\int_{1/n^2}^{\infty} \mathbb{E}[\mathbb{1}_{S_n \not \in B_\epsilon} \varphi_d(r, S_n^2)] \rho(dr) &\leq \mathbb{P} (S_n \not \in B_\epsilon) \int_{1/n^2}^{\infty}  \rho(dr) \\
&= \mathbb{P} (S_n \not \in B_\epsilon)\ \overline{\rho}(1/n^2)\\
&= o(n^{2\sigma}\ell(n^2))
%&\leq & 2\ \mathbb{P} (S_n \not \in B_\epsilon)\ n^{2\sigma} \ell(n^2)
\end{align*}

where the last line follows from the law of large numbers and Assumption \eqref{assumption:reg_var}.
Besides,
\begin{eqnarray*}
\int_{0}^{1/n^2} \mathbb{E}[\mathbb{1}_{S_n \not \in B_\epsilon} \varphi_d(r, S_n^2)] \rho(dr) &=&
 \int_{0}^{1/n^2} \mathbb{E}[\mathbb{1}_{S_n \not \in B_\epsilon} \frac{r S_n^2}{d}\varphi_{d-1}(r, S_n^2)] \rho(dr) \\
 &\leq & \int_{0}^{1/n^2} \mathbb{E}[\mathbb{1}_{S_n \not \in B_\epsilon} r S_n^2] \rho(dr) \\
&=& \mathbb{E}[\mathbb{1}_{S_n \not \in B_\epsilon} \frac{S_n^2}{n^2}]\int_{0}^{1/n^2} r n^2 \rho(dr) \\
&\leq & \mathbb{E}[\mathbb{1}_{S_n \not \in B_\epsilon} \frac{S_n^2}{n^2}]\ e\ \int_{0}^{1/n^2} r n^2 e^{-r n^2}\rho(dr) \\
&\leq & 8\mathbb{E}[\mathbb{1}_{S_n \not \in B_\epsilon} \frac{S_n^2}{n^2}]\  \frac{\sigma \Gamma(d-\sigma) }{d!} n^{2 \sigma} \ell(n^2)
\end{eqnarray*}
where the last inequality holds for $n$ large enough by Assumption \eqref{assumption:reg_var} and  Lemma $\ref{tail}$. Now, we have that
$$ \frac{S_n^2}{n^2} < \frac{1}{n} \sum\limits_{i=1}^n V_i^2$$
Since $(V_i^2)_i$ are i.i.d random variables in $\mathcal{L}_1$, we know that $(\frac{1}{n} \sum\limits_{i=1}^n V_i^2)_{n\geq 1}$ is uniformly integrable. Therefore, $(\mathbb{1}_{S_n \not \in B_\epsilon} \frac{S_n^2}{n^2})_{n\geq 1}$ is uniformly integrable. Besides, using the law of large numbers, the sequence converges almost surely, and hence in probability, to $0$. Therefore, $\lim_n \mathbb{E}[\mathbb{1}_{S_n \not \in B_\epsilon} \frac{S_n^2}{n^2}] = 0$, which concludes the proof.\\

For $\sigma = 0$, the previous computations for the upper bound give that almost surely, $\Psi_d(n) = o(\ell(n^2)) = o(\Psi(n))$. Now let $D > 1$,
$$ \mathbb{E}\sum\limits_{d \geq D} K_{n,d} = \Psi(n) - \sum\limits_{d=1}^{D-1} \Psi_d(n) \sim \ell(n^2) $$
And since $x \mapsto \sum\limits_{d\geq D} \varphi_d(1,x)$ is non-decreasing, $ (\mathbb{E}\sum\limits_{d \geq D} K_{n,d})_n $ is non-decreasing, therefore, similarly to the proof for $\sigma = 0$ for $(K_n)_n$, we find that
$$ \sum\limits_{d \geq D} K_{n,d} \sim \mathbb{E}\sum\limits_{d \geq D} K_{n,d} \sim \ell(n^2)\ \ \ a.s$$
Therefore, we finally find that
$$ \frac{K_{n,D}}{K_n} = \frac{ \sum\limits_{d \geq D} K_{n,d} - \sum\limits_{d \geq D+1} K_{n,d}}{K_n} \rightarrow 0 \ \ \ a.s$$
\end{enumerate}

\end{proof}

\section{Gibbs sampler}
\label{sec:appendixGibbs}

%\fc{check notations with the main body of the text}

As mentioned in the main text, the observed graph can be directed or undirected,  binary or count, and can have missing entries we would like to predict. Denote by $B$ the observed graph. Here we describes the steps of a Gibbs algorithm with stationary distribution
$$
p(K_n,(\widetilde r_k,\widetilde v_{1:n,k})_{k=1,\ldots,K_n},\theta \mid B).
$$
Notice that observing the full matrix $B=A$ corresponds to a weighted and directed graph with no missing entry. Let $\mathcal{I}$ denote the set of all possible edges. In the directed case, $\mathcal{I} = \{ (i,j)\ |\ 1\leq i,j \leq n \}$ and on the undirected case $\mathcal{I} = \{ (i,j)\ |\ 1\leq i \leq j \leq n \}$. We say that $(i,j)$ is not observed if we don't know the value of $A_{i,j}$. Remark that $(i,j)$ can be observed and still $A_{i,j} = 0$. Denote $\mathcal{O}$ the set of all observed entries and $\mathcal{O}^c = \mathcal{I} \smallsetminus \mathcal{O}$, the set on unobserved entry.  For all unobserved entry $(i,j) \in \mathcal{O}^c$, set $B_{i,j} = -1$

Additionally, to deal with the unknown number of active communities $K_n$, we use auxiliary slice variables $s_{i,j}$ for all $(i,j) \in \mathcal{I}$, details are given in the following paragraphs. Denote $s$ the smallest non-zero slice variable $s_{i,j}$ for $(i,j) \in \mathcal{I}$. By definition of the slice variables, $\widetilde r_k\geq s$ for all $k=1,\ldots,K_n$. Let
\begin{align*}
\overline G &=\sum_{k} r_k\delta_{v_{1:n,k}}\1{r_k\geq s}:=\sum_{k=1}^{\overline K_n} \overline r_k\delta_{\overline v_{1:n,k}}
\end{align*}
be the CRM corresponding to the set of active or inactive communities with weight $r_k\geq s$, of (almost surely finite) cardinality $\overline K_n\geq K_n$. Denote $\overline Z_{ijk}\geq 0$ the associated community interactions, and $\overline Z_{k}=(\overline Z_{ijk})$.

\subsection{Directed graph}

For each observed pair $(i,j)\in \mathcal{O}$, we define the slice variable as
\begin{equation}
s_{ij}|(\widetilde r_k,\widetilde Z_{ijk})_{k=1,\ldots,K_n}\sim \Unif\left (0, \min_{\{k|\widetilde Z_{ijk}\geq 1\}} \widetilde r_k\right )\label{eq:slice}
\end{equation}
if $A_{ij} > 0$ and $s_{ij} = 0$ otherwise. For each non observed entry $(i,j) \in \mathcal{O}^c$, we define $s_{i,j}$ by (\ref{eq:slice}) if $\{\ k\ |\ \widetilde Z_{ijk}\geq 1\} \not =  \emptyset$ and
\begin{equation}
s_{ij}|(\widetilde r_k,\widetilde Z_{ijk})_{k=1,\ldots,K_n}\sim \Unif (0, 1)
\end{equation}
otherwise.

\subsubsection{Gibbs sampler step 1 for weighted graph on observed entries}\label{gibbs:weighted}

Updating $(\widetilde Z_k)_{k=1,...,K_n}| (s, \overline G), \theta, B$ on observed entries indexes \\

We sample $(\overline Z_l)_{l=1,..,\overline K_n}$ associated to all atoms of $\overline G$ and keep only the non empty communities. For every $(i,j) \in \mathcal{O}$ such that $A_{i,j} > 0$. define the random variable $m_{ij} = \min_{\{l|\widetilde Z_{ijl}\geq 1\}} \widetilde r_l$. Then, writing the joint distribution it comes that independently for every such $(i,j)$,
$$ \mathbb{P}((\overline Z_{ijl})_{l=1,..,\overline{K}_n} | (s, G^+),\theta, B_{ij}) \propto \prod\limits_{i,j} \frac{1}{m_{ij}} \mathbb{1}_{s_{ij} < m_{ij} } \Mult((\overline Z_{ijl})_l; B_{ij},(\overline p_{ijl})_l)$$
where $\Mult$ is the multinomial distribution and $\overline p_{ijl} = \frac{r_l v_{il} v_{jl}}{\sum\limits_{t=1}^{\overline K_n} r_t v_{it} v_{jt}}$. Let $p_{ijl} = r_l v_{il} v_{jl}$. To simplify the notations, let us suppose that the atoms of $\overline G$ are in decreasing order. Remark that the indexing of $\widetilde Z$ is different from the one of $\overline Z$, the second corresponding to the one of the truncated random measure. For each observed edge $(i,j)$ independently, we can proceed in 4 phases for this step.
\begin{enumerate}
\item Sample $m_{ij}$ from the locations of $\overline G$ such that $\mathbb{P}(m_{ij} = r_L) \propto \frac{(\sum\limits_{l=1}^L p_{ijl})^{B_{ij}} - (\sum\limits_{l=1}^{L-1} p_{ijl})^{B_{ij}}}{r_L} \mathbb{1}_{s_{ij} < r_L} $.
\item For $l > L$, set $\overline Z_{ijl} = 0$
\item Sample $\overline Z_{ijL}\sim \tBin(B_{ij},\frac{p_{ijL}}{\sum\limits_{l=1}^{L} p_{ijl} })$, where $\tBin$ is the zero truncated binomial distribution
\item Sample $(\overline Z_{ij1},..,\overline Z_{ijL-1}) \sim \Mult(B_{ij}-\overline Z_{ijL},(\frac{p_{ijl}}{\sum\limits_{t=1}^{L-1} p_{ijt} })_{l \leq L-1}   )$
\end{enumerate}

\subsubsection{Gibbs sampler step 1 for unweighted graph on observed entries}\label{gibbs:unweighted}

In this setting we observe a binary matrix $B_{ij} = 1_{A_{ij} > 0}$. Then the first step of the Gibbs sampler is modified and becomes:\\

Updating $(\widetilde Z_k)_{k=1,...,K_n}| (s, \overline G), \theta, B$ on observed entries indexes\\

For each observed edge $(i,j)\in \mathcal{O}$ independently do
\begin{enumerate}
\item Sample $m_{ij}$ from the locations of $\overline G$ such that
$$\mathbb{P}(m_{ij} = r_L) \propto \frac{e^{\sum\limits_{k=1}^L p_{ijl}}-e^{\sum\limits_{k=1}^{L-1} p_{ijk}}}{r_L} \mathbb{1}_{s_{ij} < r_L} $$
Suppose $m_{ij} = r_L$
\item For $l > L$, set $\overline Z_{ijl} = 0$
\item Sample $\overline Z_{ijL}\sim \tPoisson(p_{ijL}) $, where $\tPoisson$ is the zero truncated Poisson distribution
\item For $l < L$, sample $Z_{ijl} \sim Poiss(p_{ijl}) $
\end{enumerate}

\subsubsection{Gibbs sampler step 1 on unobserved entries}\label{gibbs:prediction}

For each unobserved entry $(i,j)\in \mathcal{O}^c$, knowing $s_{ij}$, we define $L_0 = \max \{ k\ | \ r_k > s_{ij}\}$.

\begin{enumerate}
\item Draw $1_{A_{ij}=0}$, which is a Bernoulli with parameter
$$ p = \frac{1}{1 + \sum\limits_{L=1}^{L_0} \frac{e^{\sum\limits_{k=1}^L p_{ijk}}-e^{\sum\limits_{k=1}^{L-1} p_{ijk}}}{r_L}}$$
\item If $A_{i,j} \not= 0 $, then use subsection \ref{gibbs:unweighted}. Otherwise, set all counts of that entry to zero
\end{enumerate}

\subsection{Undirected graph}

In the undirected graph, we suppose that for $i \not= j$, $B_{ij} = A_{ij} + A_{ji}$ and $B_{ii} = A_{ii}$. Besides, in this setting we actually don't need to sample $\widetilde Z_{ijk}$ for all $(i,j,k)$ but only $\widetilde Z_{ijk} + \widetilde Z_{jik}$.
For each observed pair $(i,j)\in \mathcal{O}$, we define the slice variable as
\begin{equation}
s_{ij}|(\widetilde r_k,\widetilde Z_{ijk}+\widetilde Z_{jik})_{k=1,\ldots,K_n}\sim \Unif\left (0, \min_{\{k|\widetilde Z_{ijk}+\widetilde Z_{jik} \geq 1\}} \widetilde r_k\right )\label{eq:slice}
\end{equation}
if $B_{ij} > 0$ and $s_{ij} = 0$ otherwise. For each non observed entry $(i,j) \in \mathcal{O}^c$, we define $s_{ij}$ by (\ref{eq:slice}) if $\{\ k\ |\ \widetilde Z_{ijk} + \widetilde Z_{jik}\geq 1\} \not =  \emptyset$ and
\begin{equation}
s_{ij}|(\widetilde r_k,\widetilde Z_{ijk}+\widetilde Z_{jik})_{k=1,\ldots,K_n}\sim \Unif (0, 1)
\end{equation}
otherwise. Then Step 2 and 3 remain unchanged. For step 1, simply replace $p_{ijk}$ by $2 p_{ijk}$ for $i \not= j$.

\subsection{Proofs for the Gibbs sampler step 1}
\subsubsection{Weighted graph}\label{weighted:chap}
Here will give the posterior distribution of the count matrices and show that
$$(\widetilde v_k, \widetilde Z_k)_{k=1,...,K_n}| (s, \overline G), \theta, B \stackrel{dist}{=} (\widetilde v_k, \widetilde Z_k)_{k=1,...,K_n}| s, G, \theta, B$$

In order to do so, we derive the RHS posterior distribution. Let us first notice that given $G$, sampling the non zero counts and corresponding locations is equivalent to sampling $(Z_k)_k$ for $k \in \mathbb{N}$.
As stated previously, we can treat each edge $(i,j)$ independently. Therefore, we sample the sequence $(Z_{ijk})_k$. Here we suppose that the communities come with decreasing activity order. Let the random variable $L = \max \{k \ | \ Z_{ijk} > 0\}$ (supposing that the $(r_k)_k$ are decreasing). And let $p_{ijk} = r_k v_{ik} v_{jk}$

\begin{eqnarray*}
\mathbb{P}((Z_{ijk})_k| s, G, \theta, A_{ij})
&\propto& \mathbb{P}((Z_{ijk})_k|G, \theta, A_{ij}) \times \mathbb{P}(s_{ij}|(Z_{ijk})_k, G, \theta, A_{ij})\\
&\propto& \mathbb{1}_{\sum_k Z_{ijk} = A_{ij}}  \frac{A_{ij}!}{\prod\limits_{k=1}^L Z_{ijk}!} \prod\limits_{k=1}^L p_{ijk}^{Z_{ijk}}  \times \frac{1}{r_L} \mathbb{1}_{s_{ij} < r_L } \\
&\propto& \frac{(\sum\limits_{k=1}^L p_{ijk})^{A_{ij}} (1-(\frac{\sum\limits_{k=1}^{L-1} p_{ijk}}{\sum\limits_{k=1}^L p_{ijk}})^{A_{ij}}  )}{r_L} \mathbb{1}_{s_{ij} < r_L}  \\
&\times&  \mathbb{1}_{1 \leq Z_{ijL} \leq A_{ij}} \frac{A_{ij}!}{Z_{ijL}!(A_{ij}-Z_{ijL})!} \frac{ (\frac{p_{ijL}}{\sum\limits_{k=1}^L p_{ijk}})^{Z_{ijL}} (\frac{\sum\limits_{k=1}^{L-1} p_{ijk}}{\sum\limits_{k=1}^L p_{ijk}})^{A_{ij}-Z_{ijL}}}{1-(\frac{\sum\limits_{k=1}^{L-1} p_{ijk}}{\sum\limits_{k=1}^L p_{ijk}})^{A_{ij}}} \\
&\times& \mathbb{1}_{\sum_{k=1}^{L-1} Z_{ijk} = A_{ij} - Z_{ijL}}  \frac{(A_{ij}-Z_{ijL})!}{\prod\limits_{k=1}^{L-1} Z_{ijk}!} \prod\limits_{k=1}^{L-1} (\frac{p_{ijk}}{\sum\limits_{k=1}^{L-1} p_{ijk}})^{Z_{ijk}}
\end{eqnarray*}
This shows how we can sample in three steps these variables. Let us remark that the second part corresponds to the distribution of a zero truncated binomial and that the third part corresponds to the distribution of a multinomial. We also notice that only the elements of $\overline G$ are actually needed.

\subsubsection{Unweighted Graph}\label{unweighted:chap}
We proceed similarly for the unweighted graph
\begin{eqnarray*}
\mathbb{P}((Z_{ijk})_k| s, G, \theta, B_{ij} = 1)
&\propto& \mathbb{P}((Z_{ijk})_k|G, \theta, B_{ij} = 1) \times \mathbb{P}(s_{ij}|(Z_{ijk})_k, G, \theta, B_{ij} = 1)\\
&\propto&  \mathbb{1}_{Z_{ijL} \not = 0} \prod\limits_{k=1}^L \frac{p_{ijk}^{Z_{ijk}}}{Z_{ijk} !} e^{-p_{ijk} } \times e^{\sum\limits_{k=1}^L p_{ijk}} \times \frac{1}{r_L} \mathbb{1}_{s_{ij} < r_L } \\
&\propto& \frac{e^{\sum\limits_{k=1}^L p_{ijk}}-e^{\sum\limits_{k=1}^{L-1} p_{ijk}}}{r_L} \mathbb{1}_{s_{ij} < r_L}  \\
&\times&  \mathbb{1}_{Z_{ijL} \not= 0} \frac{1}{1-e^{-p_{ijL}}} \frac{p_{ijk}^{Z_{ijL}}}{ Z_{ijL}! } e^{-p_{ijL} } \\
&\times& \mathbb{1}_{Z_{ij(L+1)},... = 0} \prod\limits_{k=1}^{L-1} \frac{p_{ijk}^{Z_{ijk}}}{Z_{ijk} !} e^{-p_{ijk} }
\end{eqnarray*}

\subsubsection{Prediction}\label{prediction:chap}
Here we show how to update the missing entries we try to predict. Let us recall that for a predicted count, if it is positive, we define the slice variable as previously. However, if the count is equal to zero, then the slice variable is simply uniform over $[0,1]$. Now let $L_0 = \max \{ k\ | \ r_k > s_{ij}\}$

\begin{eqnarray*}
\mathbb{P}((Z_{ijk})_k| s, G, \theta)
&\propto& \mathbb{P}((Z_{ijk})_k|G, \theta) \times \mathbb{P}(s_{ij}|(Z_{ijk})_k, G, \theta)\\
&\propto& \mathbb{1}_{Z_{ij(L_0+1)},... = 0}\ e^{\sum\limits_{k=1}^{L_0} p_{ijk}} \prod\limits_{k=1}^{L_0} \frac{p_{ijk}^{Z_{ijk}}}{Z_{ijk} !} e^{-p_{ijk} }  \times (\mathbb{1}_{A_{ij} \not= 0} \frac{1}{r_L} \mathbb{1}_{s_{ij} < r_L } + \mathbb{1}_{A_{ij} = 0})
\end{eqnarray*}
Now let
$$ f((Z_{ijk})_k) = \mathbb{1}_{Z_{ij(L_0+1)},... = 0}\  e^{\sum\limits_{k=1}^{L_0} p_{ijk}} \prod\limits_{k=1}^{L_0} \frac{p_{ijk}^{Z_{ijk}}}{Z_{ijk} !} e^{-p_{ijk} } \times (\mathbb{1}_{A_{ij} \not= 0} \frac{1}{r_L} \mathbb{1}_{s_{ij} < r_L } + \mathbb{1}_{A_{ij} = 0})$$

Using \ref{unweighted:chap}, it comes that
$$ \mathbb{E} ( f((Z_{ijk})_k)\ |\ A_{ij} \not = 0) = e^{\sum\limits_{k=1}^{L_0} p_{ijk}} \sum\limits_{L=1}^{L_0} \frac{e^{\sum\limits_{k=1}^L p_{ijk}}-e^{\sum\limits_{k=1}^{L-1} p_{ijk}}}{r_L} $$

Besides,
$$ \mathbb{E} ( f((Z_{ijk})_k)\ |\ A_{ij} = 0) = f(0) = 1$$

Therefore, here we proceed in two steps, first we sample the binomial $\mathbb{1}_{A_{ij} = 0}$ with parameter
$$ p = \frac{1}{1 + e^{\sum\limits_{k=1}^{L_0} p_{ijk}} \sum\limits_{L=1}^{L_0} \frac{e^{\sum\limits_{k=1}^L p_{ijk}}-e^{\sum\limits_{k=1}^{L-1} p_{ijk}}}{r_L}}$$

Then, conditioning on the event $A_{ij} \not= 0$, we use \ref{unweighted:chap} to proceed.

\subsection{Proof for the Gibbs step 2}

Here we show how we can update the parameters $\theta = (\kappa,\sigma,\tau,\alpha,\beta)$ using a Metropolis-Hastings update. First, let us derive the posterior distribution of the hyperparameters.

% Posterior of the hyperparameters
We write $G = G' + \sum\limits_{c = 1}^K \tilde{r}_c \delta_{\tilde{v}_c}$ where $G'$ is the non observed part. And we note $\overline{G'}$ the restriction of $G'$ to the locations which intensity is larger than $\min s$.

\begin{eqnarray*}
p(\theta | (s,\overline G), \widetilde{v},\widetilde{Z}) &\propto& p(\ \theta \ | \sum\limits_{c = 1}^K \tilde{r}_c \delta_{\tilde{v}_c}, s, \overline{G'}, \widetilde{Z}) \\
&\propto& p(\ \theta \ , \sum\limits_{c = 1}^K \tilde{r}_c \delta_{\tilde{v}_c}, s, \overline{G'}, \widetilde{Z}) \\
&\propto& p(\ \theta \ , \sum\limits_{c = 1}^K \tilde{r}_c \delta_{\tilde{v}_c}, \widetilde{Z}) \  p(\ s, \overline{G'} | \ \theta \ , \sum\limits_{c = 1}^K \tilde{r}_c \delta_{\tilde{v}_c}, \widetilde{Z})
\end{eqnarray*}

Now let us derive consider the first part
\begin{eqnarray*}
p(\ \theta \ , \sum\limits_{c = 1}^K \tilde{r}_c \delta_{\tilde{v}_c}, \widetilde{Z}) &\propto& p(\theta) \ p(\sum\limits_{c = 1}^K \tilde{r}_c \delta_{\tilde{v}_c} | \theta)\ p ( \widetilde{Z}| \sum\limits_{c = 1}^K \tilde{r}_c \delta_{\tilde{v}_c}) \\
&\propto& p(\theta) p(\sum\limits_{c = 1}^K \tilde{r}_c \delta_{\tilde{v}_c} | \theta) \\
&\propto& p(\theta) e^{-\Psi(n)} \prod\limits_{c=1}^K \rho_{\kappa,\sigma,\tau}(\tilde{r}_c) f_{\alpha,\beta}(\tilde{v}_c)
\end{eqnarray*}

Now let us consider the second part
\begin{eqnarray*}
p(\ s, \overline{G'} \ | \ \theta \ , \sum\limits_{c = 1}^K \tilde{r}_c \delta_{\tilde{v}_c}, \widetilde{Z}) &\propto&
p(s | \sum\limits_{c = 1}^K \tilde{r}_c \delta_{\tilde{v}_c}, \widetilde{Z})\ p(\overline{G'} | \theta, s ) \\
&\propto& p(\overline{G'} | \theta, s ) \\
&\propto& e^{-\Psi'(\min s,n)} \prod_k \rho_{\kappa,\sigma,\tau}(r'_k) f_{\alpha,\beta}(v'_k)
\end{eqnarray*}
where $(r'_k)$ and $(v'_k)$ are respectively the intensities and locations of $\overline{G'}$

% MH update

Let
$$\pi(\theta) = e^{-\Psi (n) -\Psi'(\min s,n)} \prod\limits_{t=1}^T \mu_{\theta}(r_t,v_t),$$
where we are taking the product over the $T$ atoms and jumps of $\overline{G}$ and
$$\Psi' (s,n) = \int_{r > s,v} e^{-r |v|^2 } \prod\limits_{i=1}^n (f(v_i) dv_i)\ \rho(r) dr.$$
The posterior satisfies $p(\theta | (s,\overline{G})) \propto p(\theta) \pi(\theta)$. With our particular choice of distribution of the CRM, the multivariate integrals are reduced to one dimensional integrals, which makes the algorithm tractable. Indeed, we find that

$$\Psi (n) + \Psi'(\min s,n) = \frac{\kappa}{\sigma} \int_0^{+\infty} (\tau+\varsigma^2)^\sigma [ \frac{\sigma}{\Gamma(1-\sigma)} \Gamma(-\sigma,(\tau+\varsigma^2)\min s ) + 1] f_{n\alpha,\beta}(\varsigma)d\varsigma - \frac{\kappa \tau^\sigma}{\sigma} $$

\[
    \Psi (n)=
\begin{cases}
    \frac{\kappa}{\sigma} \int_0^{+\infty} (\tau+\varsigma^2)^\sigma  f_{n\alpha,\beta}(\varsigma)d\varsigma - \frac{\kappa \tau^\sigma}{\sigma}   ,& \text{if } \sigma > 0 \text{ or } \sigma < 0\\
    \kappa \int_0^{+\infty} \log(\tau+\varsigma^2) f_{n\alpha,\beta}(\varsigma)d\varsigma - \kappa \log \tau   ,& \text{if } \sigma =  0
\end{cases}
\]
and
\[
    \Psi'(\min s,n)=
\begin{cases}
\frac{\kappa}{\Gamma (1-\sigma)} \int_0^{+\infty} (\tau+\varsigma^2)^\sigma  \Gamma(-\sigma,(\tau+\varsigma^2)\min s ) f_{n\alpha,\beta}(\varsigma)d\varsigma,   & \text{if } \sigma > 0 \text{ or } \sigma < 0\\
    \kappa \int_0^{+\infty} \int_{(\tau+\varsigma^2)\min s}^{+\infty} \frac{e^{-r}}{r} dr f_{n\alpha,\beta}(\varsigma)d\varsigma   ,& \text{if } \sigma =  0
\end{cases}
\]

We use the following priors:
\begin{eqnarray*}
1-2\sigma &\sim & Gamma(a_{\sigma},b_{\sigma}) \\
\kappa &\sim & Gamma(a_{\kappa},b_{\kappa}) \\
\tau &\sim & Gamma(a_{\tau},b_{\tau}) \\
\alpha &\sim & Gamma(a_{\alpha},b_{\alpha}) \\
\beta &\sim & Gamma(a_{\beta},b_{\beta})
\end{eqnarray*}
And proposals
\begin{eqnarray*}
1-2\tilde{\sigma} | \sigma &\sim &  Lognormal(\log (1-2\sigma), \Sigma_\sigma) \\
\tilde{\kappa} | \kappa &\sim & Lognormal(\log \kappa, \Sigma_\kappa) \\
\tilde{\tau}|\tau &\sim & Lognormal(\log \tau, \Sigma_\tau) \\
\tilde{\alpha} | \alpha &\sim & Lognormal(\log \alpha, \Sigma_\alpha) \\
\tilde{\beta}|\beta &\sim & Lognormal(\log \beta, \Sigma_\beta)
\end{eqnarray*}
We find that
\begin{eqnarray*}
\log \frac{p(\tilde{\theta}) q(\theta|\tilde{\theta})}{p(\theta) q(\tilde{\theta} | \theta)} &=& a_\sigma \log \frac{1-2\tilde{\sigma}}{1-2\sigma} + 2b_\sigma (\tilde{\sigma}-\sigma) \\
&+& a_\kappa \log \frac{\tilde{\kappa}}{\kappa} - b_\kappa (\tilde{\kappa}-\kappa) \\
&+& a_\tau \log \frac{\tilde{\tau}}{\tau} - b_\tau (\tilde{\tau}-\tau) \\
&+& a_\alpha \log \frac{\tilde{\alpha}}{\alpha} - b_\alpha (\tilde{\alpha}-\alpha) \\
&+& a_\beta \log \frac{\tilde{\beta}}{\beta} - b_\beta (\tilde{\beta}-\beta)
\end{eqnarray*}
And
\begin{eqnarray*}
\log \frac{\pi(\tilde{\theta})}{\pi(\theta)} &=& \Psi_\theta(n) + \Psi'_\theta(\min s, n) - \Psi_{\tilde{\theta}}(n) - \Psi'_{\tilde{\theta}}(\min s, n) \\
&+& T \log \frac{\tilde{\kappa}}{\kappa} - T \log \frac{\Gamma(1-\tilde{\sigma})}{\Gamma(1-\sigma)} -n T \log \frac{\Gamma(\tilde{\alpha})}{\Gamma(\alpha)} \\
&-& (\tilde{\sigma} - \sigma) \sum_t \log r_t - (\tilde{\tau} - \tau) \sum_t r_t \\
&+& (\tilde{\alpha}-\alpha) \sum_{t,i} \log v_{t,i} - (\tilde{\beta} - \beta) \sum_{t,i} v_{t,i} \\
&+& n T (\tilde{\alpha} \log \tilde{\beta} - \alpha \log \beta )
\end{eqnarray*}

\subsection{Sampling from the inhomogeneous CRM}\label{sec:inhomo}
In this section we show how we can sample from the inhomogeneous CRM $G'$ with measure:
$$ \mu'(dr,dv) = e^{-r (\sum\limits_{i=1}^n v_{i})^2} \rho_{\kappa,\tau,\sigma}(r)\ [\prod\limits_{i=1}^n f_{\alpha,\beta}(v_i)dv_i]\ dr$$
Let us recall that $f_{\alpha,\beta}$ is the gamma pdf and $\rho_{\kappa,\tau,\sigma}$ the GGP intensity. From Section 4, we know that if we make the following change of variables
$(v_1,..,v_n) \mapsto (\varsigma =\sum_i v_i, \nu_1 = v_1/s,...,\nu_n = v_n/s)$, we get
$$ \mu'(dr,d\varsigma,d\nu) = e^{-r \varsigma^2} \rho_{\kappa,\tau,\sigma}(r)\ f_{n\alpha,\beta}(\varsigma)d\varsigma\ dr\ \text{Dir}(d\nu_1,...,d\nu_n; \pmb{\alpha})$$
Hence, we can sample independently $(r,\varsigma)$ and $\nu$. From one hand, $(\nu_1,..,\nu_n)$ is sampled from a Dirichlet distribution with parameter $\pmb{\alpha} = (\alpha,..,\alpha)$.
On the other hand, the total sum $\varsigma$ and the intensity $r$ are sampled from
\begin{eqnarray*}
\mu(r,\varsigma) &=&  e^{-r \varsigma^2} \rho_{\kappa,\tau,\sigma}(r)\ f_{n\alpha,\beta}(\varsigma) \\
&=& \frac{\kappa}{\Gamma (1-\sigma)} e^{-r (\varsigma^2+\tau)} r^{-1-\sigma} f_{n\alpha,\beta}(\varsigma)
\end{eqnarray*}
Now, to reduce the problem to sampling from a homogeneous CRM, let us consider the change of variable $(r,\varsigma) \mapsto (\overline{r}= r[\tau+\varsigma^2], s)$ which determinant is $\tau+\varsigma^2$. We find finally that
\begin{eqnarray*}
\mu(\overline{r},\varsigma) = \frac{\kappa}{\Gamma (1-\sigma)}  e^{-\overline{r}}\  \overline{r}^{-1-\sigma}\ (\varsigma^2+\tau)^{\sigma} f_{n\alpha,\beta}(\varsigma)
\end{eqnarray*}
Besides, since $\overline{r} \geq r \tau\ \ \forall (i,j)$, we only need to sample the points such that $\overline{r} \geq \tau \min s$. Therefore, since $\tau > 0$, we sample a finite number of atoms. Then, we only keep the points such as $r > \min s$.
Finally, let us notice that in our setting, even with $\sigma \leq 0$
$$ \int_\varsigma (\varsigma^2+\tau)^{\sigma}f_{n\alpha,\beta}(\varsigma) d\varsigma = \frac{\sigma}{\kappa} \Psi(n) + \tau^{\sigma}$$
Therefore, we first sample the jumps from the levy measure
$$\rho(\overline{r}) =  \frac{\sigma \Psi(n) + \kappa\tau^{\sigma}}{\Gamma (1-\sigma)}  e^{-\overline{r}}\  \overline{r}^{-1-\sigma} \mathbb{1}_{\overline{r}\geq \tau \min s}$$
using adaptive thinning \citep{Favaro2013}. Then, we sample $\varsigma$ with pdf $\propto (\varsigma^2+\tau)^{\sigma}f_{n\alpha,\beta}(\varsigma)$ using rejection sampling.

\end{document}